\documentclass{article}

\usepackage{arxiv}
\usepackage{amsmath, amssymb, amsthm}
\usepackage{graphicx,color,rotating}
\usepackage[utf8]{inputenc}
\usepackage[square,sort,comma,numbers]{natbib}
\usepackage[ruled,linesnumbered,vlined]{algorithm2e}
\usepackage{capt-of}
\usepackage{setspace}
\usepackage{xspace}
\usepackage{soul}
\usepackage{paralist}
\usepackage{url}
\usepackage[hidelinks]{hyperref}

\theoremstyle{plain}
\newtheorem{theorem}{Theorem}

\newtheorem{lemma}         [theorem]{Lemma}

\newtheorem{corollary}     [theorem]{Corollary}

\newtheorem*{remark} {Remark}

\long\def\commhide #1\commhideend{}

\newcommand{\ID}{ID\xspace}
\newcommand{\IDs}{IDs\xspace}

\newcommand{\onencc}{{\mathsf{NCC}_1}}
\newcommand{\zeroncc}{{\mathsf{NCC}_0}}
\newcommand{\val}{{\mathsf{value}}}

\newcommand{\tok}{{\mathsf{token}}}
\newcommand{\gid}{{\mathsf{gid}}}

\newcommand{\elev}{\mathsf{elev}}

\renewcommand{\L}{\mathcal{L}}
\newcommand{\hh}{\textsf{Havel-Hakimi }}

\newcommand{\blue}[1]{\textcolor{blue}{#1}}
\newcommand{\magenta}[1]{\textcolor{magenta}{#1}}
\newcommand{\green}[1]{\textcolor{green}{#1}}

\newcommand{\tO}{\widetilde O}
\newcommand{\plog}[1]{\mathrm{polylog}(#1)}

\def\dnsparagraph#1{\par\vspace{3pt}\noindent\textbf{#1}}
\def\dnsparagraphem#1{\par\vspace{3pt}\noindent\textit{#1}}

\ifdefined\ShowComment
\def\john#1{\marginpar{$\leftarrow$\fbox{J}}\footnote{$\Rightarrow$~{\sf #1 \magenta{--John}}}}
\def\sumathi#1{\marginpar{$\leftarrow$                                                                      \fbox{S}}\blue{ $\bigstar$}\footnote{$\Rightarrow$~{\sf #1 \blue{--Sumathi}}}}
\def\keerti#1{\marginpar{$\leftarrow$\fbox{K}}\blue{ $\bigstar$}\footnote{$\Rightarrow$~{\sf #1 \blue{--Keerti}}}}
\def\avi#1{\marginpar{$\leftarrow$\fbox{A}}\green{ $\bigstar$}\footnote{$\Rightarrow$~{\sf #1 \green{--Avi}}}}
\def\suman#1{\marginpar{$\leftarrow$\fbox{Su}}\green{ $\bigstar$}\footnote{$\Rightarrow$~{\sf #1 \cyan{--Suman}}}}
\def\david#1{\marginpar{$\leftarrow$\fbox{D}}\green{ $\bigstar$}\footnote{$\Rightarrow$~{\sf #1 \teal{--David}}}}
\newcommand{\comment}[1]{{\blue{{\em #1}}}}
\else
\def\john#1{}
\def\sumathi#1{}
\def\keerti#1{}
\def\avi#1{}
\def\suman#1{}
\def\david#1{}
\newcommand{\comment}[1]{}
\fi

\newboolean{short}
\setboolean{short}{false}
\newcommand{\shortOnly}[1]{\ifthenelse{\boolean{short}}{#1}{}}
\newcommand{\onlyShort}[1]{\ifthenelse{\boolean{short}}{#1}{}}
\newcommand{\longOnly}[1]{\ifthenelse{\boolean{short}}{}{#1}}
\newcommand{\onlyLong}[1]{\ifthenelse{\boolean{short}}{}{#1}}
\newcommand{\shortLong}[2]{\ifthenelse{\boolean{short}}{#1}{#2}}
\newcommand{\longShort}[2]{\ifthenelse{\boolean{short}}{#2}{#1}} 
\newboolean{IPDPS}
\setboolean{IPDPS}{true}
\newcommand{\onlyIPDPS}[1]{\ifthenelse{\boolean{IPDPS}}{#1}{}}

\def\blackslug{\hbox{\hskip 1pt \vrule width 4pt height 8pt
    depth 1.5pt \hskip 1pt}}
\def\QED{\quad\blackslug\lower 8.5pt\null\par}
\def\cT{{\cal T}}
\def\cD{{\cal D}}

\def\cTvd{\cT}
\def\edgeconn{\textsf{Conn}}

\title{Distributed Graph Realizations}

\begin{document}

\makeatletter
\newcommand{\linebreakand}{%
  \end{@IEEEauthorhalign}
  \hfill\mbo

  \mbox{}\hfill\begin{@IEEEauthorhalign}
}
\makeatother
\author{
John Augustine\\
  Computer Science and Engineering\\
  Indian Institute of Technology Madras\\
  \texttt{augustine@cse.iitm.ac.in} \\
   \And
 Keerti Choudhary \\
 Computer Science and Applied Mathematics\\
  Weizmann Institute of Science\\
  \texttt{keerti.choudhary@weizmann.ac.il} \\
  \And
 Avi Cohen \\
  Computer Science and Applied Mathematics\\
  Weizmann Institute of Science\\
  \texttt{avi.chohen@weizmann.ac.il} \\
    \And
 David Peleg\\
 Computer Science and Applied Mathematics\\
  Weizmann Institute of Science\\
  \texttt{david.peleg@weizmann.ac.il} \\
    \And
Sumathi Sivasubramaniam \\
  Computer Science and Engineering\\
  Indian Institute of Technology Madras\\
  \texttt{sumathi@cse.iitm.ac.in} \\
    \And
 Suman Sourav \\
  Computer Science\\
  National University of Singapore\\
  \texttt{sourav@comp.nus.edu.sg} \\
}

\maketitle

\begin{abstract}
    We study graph realization problems from a distributed perspective. The problem is naturally applicable to the distributed construction of overlay networks that must satisfy certain degree or connectivity properties, and we study it in the node capacitated clique (NCC) model of distributed computing, recently introduced for representing peer-to-peer networks. 
    
    We focus on two central variants, \emph{degree-sequence realization} and \emph{minimum threshold-connectivity} realization. 
    In the degree sequence problem, each node $v$ is associated with a degree $d(v)$, and the resulting degree sequence is realizable if it is possible to construct an overlay network in which the degree of each node $v$ is $d(v)$. The minimum threshold-connectivity problem requires us to construct an overlay network that satisfies connectivity constraints specified between every pair of nodes. 
    
    Overlay network realizations can be either explicit or implicit.
    Explicit realizations require both endpoints of any edge in the realized graph to be aware of the edge. In implicit realizations, on the other hand, at least one endpoint of each edge of the realized graph needs to be aware of the edge. 
    
    The main realization algorithms we present are the following.
    \begin{itemize}
        \item  An $\tilde{O}(\min\{\sqrt{m},\Delta\})$ time algorithm for implicit realization of a degree sequence. Here, $\Delta = \max_v d(v)$ is the maximum degree and $m = (1/2) \sum_v d(v)$ is the number of edges in the final realization.
        \item $\tilde{O}(\Delta)$ time algorithm for an explicit realization of a degree sequence. We first compute an implicit realization and then transform it into an explicit one in  $\tilde{O}(\Delta)$ additional rounds.
        \item An $\tilde{O}(\Delta)$ time algorithm for the threshold connectivity problem that obtains an explicit solution and an improved $\tilde{O}(1)$ algorithm for implicit realization when all nodes know each other's IDs. These algorithms yield 2-approximations w.r.t. the number of edges.
    \end{itemize}
    We complement our upper bounds with lower bounds to show that the above algorithms are tight up to factors of $\log n$. Additionally, we provide algorithms for realizing trees (including a procedure for obtaining a tree with a minimal diameter), an $\tilde{O}(1)$ round algorithm for approximate degree sequence realization and finally an $O(\log^2 n)$ algorithms for degree sequence realization in the non-preassigned case namely, where the input degree sequence may be permuted among the nodes.

\end{abstract}

 \section{Introduction}

\subsection{Background and motivation}

{\em Graph Realization problems}, which deal with constructing graphs that satisfy certain specified properties, have been studied quite extensively for over half a century, focusing on problems related to realizing graphs with specified degrees~\cite{H55,EG60,Ha62}, as well as other properties, like connectivity and flow \cite{GH61,FC70,F92,F94} or eccentricities \cite{B76,L75}.

The most prominent realization problems deal with degree sequences. A (typically non-increasing) sequence of non-negative numbers $D = (d_1, d_2, \ldots, d_n)$
is said to be {\em realizable} or {\em graphic} if there is a graph on $n$ nodes whose sequence of degrees matches $D$. By the {\em handshaking lemma}, dating back to Euler's work on the Konigsburg bridges problem~\cite{Euler1736}, we know that if $D$ is graphic, then $\sum_i d_i$ must be even. In 1960, Erd\"os and Gallai~\cite{EG60} gave a complete characterization, showing that $D$ is graphic iff $\sum_{i=1}^k d_i \le k(k-1) + \sum_{i=k+1}^n \min(d_i, k)$ for every $k \in [1,n]$.
Havel~\cite{H55} and Hakimi~\cite{Ha62} independently gave a recursive algorithm that can determine whether a given $D$ is graphic and compute a realizing graph when it is graphic. Being constructive, their 
method
has played a crucial role in many subsequent studies, and ours is no exception.

Over the last two decades, peer-to-peer (P2P) networks have developed as a versatile and effective platform for cooperative distributed computations. Research on P2P has lead to ideas that have  become crucial in a variety of contexts, ranging from fully decentralized applications like blockchain networks to more controlled contexts like Akamai's network services~\cite{akamai}. 
Overlay construction is an important P2P component, involving the formation of new links -- so called overlay links -- that comprise an overlay network tailored to benefit P2P applications. In the typical scenario, starting from some basic network state, the nodes in a P2P network must interact with each other in a fully decentralized manner and form an overlay network to be used for specific purposes. 

The constructed overlay network $G$ is often required to possess certain desirable properties. A common requirement is that $G$ be of bounded degree, so that the overhead for network formation and maintenance at each node is bounded. Additionally, one can envision a variety of other desired properties that the overlay $G$ should possess, like bounded diameter, well-connectedness, flow guarantees, tolerance to both benign and malicious failures, and so on.

Note, however, that such overlay constructions can be viewed as (distributed) graph realization problems. This natural connection makes it plausible that ideas from graph realization theory may lend interesting new techniques allowing us to build better overlay networks. Conversely, the endeavor to build useful overlay networks is expected to pose new theoretical challenges that are likely to enrich graph realization theory. We hope that our work will initiate this new synergy between these two areas that, to the best of our knowledge, has not been formally explored in the past. 

Towards the goal of formulating and studying distributed graph realization problems, we employ the {\em node capacitated clique (NCC)} model~\cite{AGGHSKL19} that captures several key aspects of P2P networks. In this model, we have $n$ nodes $V$ with unique identifiers called \IDs. Any node $u$ can send messages of bounded size to any other node $v$ provided $u$ knows $v$'s \ID; we can think of $v$'s \ID as its IP address. In this sense, the NCC model is somewhat similar to the {\em congested clique (CC)}
model~\cite{LSPP05,KS18}. However, in the interest of being useful in the P2P context, NCC limits each node to send and receive a bounded number of messages, which, interestingly, makes NCC quite distinct from CC.

In the first paper that introduced NCC~\cite{AGGHSKL19}, node \IDs were assumed to be common knowledge. This knowledge of all other \IDs corresponds roughly to the $KT_1$ version of the CONGEST model (cf. \cite{P00}) and for this reason we call it the $\onencc$ model. One may argue that $\onencc$ is somewhat impractical in the P2P context because  peers are highly unlikely to have so much knowledge about other nodes. 
To address this, we introduce a variant of NCC called $\zeroncc$, corresponding to the $KT_0$ version of the CONGEST model, which limits the number of \IDs known to each node. The knowledge available to the nodes at the beginning of the execution is represented by a directed graph on $V$, called the {\em initial knowledge graph} $G_k$, in which the (out) neighbors of each $v$ are the nodes whose \IDs were known to $v$ at the start. In our case, we make the simplifying assumption that $G_k$ is a path of $n$ nodes. Later, we show how our algorithms can be adapted for other initial knowledge graphs. In fact, we show that our algorithms be run on any low degree initial knowledge graph. Since $\zeroncc$ is more restrictive, algorithms designed for it will also work in $\onencc$. 
Consequently,
we focus on algorithms for $\zeroncc$ (unless stated otherwise). 
Further details of the model are described in Section \ref{sec:model}.

\dnsparagraph{Problem Statements.~}
We say that an overlay graph $G=(V,E)$ is constructed if, for every $e = (u,v) \in E$, at least one of the endpoints is aware of the \ID of the other and also aware that $e \in E$. We say that the overlay  graph is {\em explicit} if, for every edge in the graph, both endpoints are aware of the edge. Otherwise, the overlay is said to be {\em implicit}. In this paper, we focus on distributed realization problems in which, from a given initial knowledge graph and some other required input parameters, we are to construct an overlay graph that satisfies certain requirements. We study both explicit and implicit versions of the following two realization problems.

\dnsparagraphem{Degree Realization:~}
Each node $v_i$ in the distributed network knows its required degree $d(v_i)=d_i$. The goal is to construct a realizing graph (if one exists).
Formally, our input is a vector $D = (d_1, d_2, \ldots, d_n)$ such that each $d_i$ is only known to the corresponding node $v_i$. The required output is an overlay graph in which every $v_i$ has degree $d_i$ if $D$ is realizable; otherwise, at least one node outputs {\sf Unrealizable}. We also explore degree realization in the \emph{non-preassigned} setting. In this case, each node $v_i$ is given a degree $d_i$, but the required output is an overlay graph whose degree sequence is only required to be a permutation of $D$. 
%
%

\dnsparagraphem{Connectivity Threshold Realization:~}
The local edge connectivity of two nodes $u$ and $v$, denoted $\edgeconn_G(u,v)$, captures the minimum number of edge disjoint paths required between the nodes $u$ and $v$. In the connectivity realization problem, each node $v$ in the distributed network is provided with a vector specifying the required minimum local edge connectivity (denoted by  $\sigma(u,v)$) to every other node $u$.\footnote{Here $\sigma$ is a symmetric relation, that is, $\sigma(u,v)=\sigma(v,u)$.}
The goal is then to compute an overlay graph $G$ with as few edges as possible so that any two nodes $u,v$ in $G$ satisfy the \emph{edge-connectivity} relation $\edgeconn_G(u,v) \geq \sigma(u,v)$.

For this problem, we primarily focus on an approximate solution. In particular, we ensure that the number of edges in the overlay network is larger by at most twice that of the optimal realization.

\subsection{Our Contributions}

We present a number of new algorithms. Unless stated otherwise, these are randomized (Las Vegas) algorithms, whose running time bounds hold with high probability (w.h.p.)\footnote{An event such as running time staying within some bound is said to hold {\em with high probability} if its probability is at least $1-n^{-c}$ for any constant $c$  independent of $n$, but may depend on the constants used within the algorithm.}. For the following results, $m$ is the number of edges, and $\Delta$ denotes the maximum degree in the given degree sequence.
\begin{compactenum}
\item For the distributed degree sequence realization problem, we provide an $\tilde{O}(\min(\sqrt{m}, \Delta))$ time\footnote{We use $\tilde{O}$ and $\tilde{\Omega}$ as asymptotic notations that hide $\text{polylog }n$ factors.} algorithm that produces an implicit realization of the given graphic degree sequence. We then adapt this algorithm into an $\tilde{O}(\Delta)$ time algorithm for explicit realizations. 
We also study {\em tree realizations}, and present an algorithm for implicit realization in $\tilde{O}(1)$ rounds. Furthermore, we optimize the diameter of the realized tree.

\item We give an $\tilde{O}(1)$ round algorithm for implicit connectivity threshold realizations in the $\onencc$ model that uses at most twice the optimal number of edges 
needed for satisfying the connectivity threshold requirements. 
For the $\zeroncc$ model, we give an $\tilde{O}(\Delta)$ round explicit realization algorithm.

\item All our algorithms are tight up to factors of $\log n$ in the $\zeroncc$ model. Specifically, we show that for implicit realizations of degree sequences there are instances that require at least $\tilde{\Omega}(\sqrt{m})$ rounds and other instances that require $\tilde{\Omega}(\Delta)$ rounds in $\zeroncc$. 
In comparison, we also show that every instance of explicit realization requires at least $\tilde{\Omega}(\Delta)$ rounds in $\zeroncc$. 

\item To facilitate the design of algorithms in $\zeroncc$, we provide some algorithmic primitives that may be of independent interest. First, we show that the nodes can be arranged in the form of a balanced binary tree in deterministic $O(\log n)$ time. Furthermore, they can be rearranged to form a path sorted according to some parameter known locally to each node in deterministic $O(\plog{n})$ rounds. Finally, we show how a series of  primitives presented in~\cite{AGGHSKL19} for $\onencc$ can be adapted to work in $\zeroncc$.

\item Finally, for the distributed degree sequence problem, in the non-preassigned case (where it is only required that the final realization be a permutation of the assigned sequence) we provide an $(1+\epsilon)$ solution that runs in $O(\log^2 n)$ number of rounds in $\onencc$. 
\end{compactenum}

\subsection{Related Work}
A variety of graph realization problems have been studied in the literature. For the problem of realizing degree sequences, Havel and Hakimi~\cite{H55,Ha62} independently came up with the recursive algorithm that forms the basis for our distributed algorithm. Non-centralized versions of realizing degree sequences have also been studied, albeit to a lesser extent.  Arikati and Maheshwari~\citep{AM1996realizing} provide an efficient technique to realize degree sequences in the PRAM model. To the best of our knowledge, graph realization problems  have not been explored in the distributed setting.

Other graph realization problems studied are eccentricities~\citep{B76,L75}, connectivity~\citep{F94,FC70}, degree interval sequences (cf.\citep{BCPR-interval:19} and references therein), and more (cf. ~\citep{BCPR18survey,BCPR19happiness}).

Overlays in distributed settings have been well studied, as they provide structure and stability which are best exemplified by structures such as Chord~\citep{SML03}, CAN~\citep{RFHKS01}, and Skip Graphs~\citep{AS07}. They can also be used to handle dynamism and faults. For example,  Fiat and Saia~\citep{AS02} introduced an overlay using butterfly networks that is tolerant of faults. This was used in~\citep{AGGHSKL19,AS18} to provide a low diameter structure that is also easily addressable.  Detailed  surveys of various overlays and their properties can be found in~\cite{M15,LCPSL05}. 

Our work uses the NCC model~\cite{AGGHSKL19} for P2P networks. This model is similar to the CC 
model~\citep{LSPP05}, introduced well over a decade ago and studied extensively since then (see~\cite{KS18} and references therein).

\dnsparagraph{Organization.~}
In Section~\ref{sec:model} we formally define the NCC model of communication. Then, to lay the groundwork for our algorithms, Section~\ref{s:prelim} presents a series of primitives along with a brief description of the Havel-Hakimi algorithm. Our main contributions are in sections \ref{sec:deg-realiz-general} through \ref{sec:lb}. Section~\ref{sec:deg-realiz-general} presents algorithms for realizing general graphs given graphic degree sequences. The special case of trees is treated in Section~\ref{s:trees}. Section~\ref{section:Connectivity} studies the problem of connectivity threshold realization. Our lower bounds are presented in Section~\ref{sec:lb}. In Section~\ref{sec:non-preassigned}, we present our approach to the non-preassigned case for degree sequence realization, where the degrees aren't assigned to specific nodes.
Throughout this paper, we assume that the initial knowledge graph $G_k$ is a directed path on $n$ nodes. In Section~\ref{sec:equivalence} we introduce the notion of graph equivalences and show how to eliminate the restriction that the initial knowledge graph is a path. Finally, our conclusions are offered in Section~\ref{sec:con}.


\section{The NCC model of communication} \label{sec:model}

We consider the node capacitated clique model (NCC) recently introduced in 
\cite{AGGHSKL19}. The NCC comprises $n$ nodes $V = \{v_1, v_2, \ldots, v_n\}$ that can communicate with each other via synchronous message passing. Each node is uniquely identified by an \ID 
drawn from $[1, n^c]$ for some  fixed $c \gg 1$. A node can send at most $O(\log n)$ messages of size $O(\log n)$ bits each per round. However, in order for a node $u$ to send a message to another node $v$, $u$ must know the \ID of $v$. (Intuitively, $v$'s \ID can be viewed as its IP address.)

Since inter-node communication is significantly slower than local computation at the nodes, our focus is on minimizing the {\em round complexity} (a.k.a. {\em time} complexity), which measures the number of synchronous communication rounds it takes an algorithm to terminate. Hence, the NCC model allows nodes to perform an unbounded amount of local computation per round. However, we emphasize that in our algorithms, the computation times are bounded by some polynomial in $n$.  

We distinguish two variants of the NCC setting, $\zeroncc$ and $\onencc$, based on the \IDs initially known to the nodes. In $\onencc$, which matches the NCC model introduced in \cite{AGGHSKL19}, all nodes have full knowledge of each others' \IDs. Thus, w.l.o.g., we can assume that the \IDs are in $[1,n]$. This version of the NCC is similar to the $KT_1$ variant of the CC model~\cite{LSPP05}, except for the bound on the number of messages that can be sent/received at a node in every round. In $\zeroncc$, on the other hand, 
each node only knows the \IDs of a few other nodes. Formally, each node knows only the \IDs of its neighbors in some directed  \emph{initial knowledge graph} $G_k = (V, E_k)$, such that a pair $(u,v) \in E_k$ iff $u$ knows $v$'s \ID at the beginning. For concreteness, in this paper we assume that $G_k$ is a directed path consisting of the $n$ nodes arranged in some arbitrary order. In Section~\ref{sec:equivalence}, we show how our algorithms can be adapted to any low degree initial knowledge graph.
\begin{remark}
Any algorithm that can be executed in the $\zeroncc$ model can be executed in the $\onencc$ model without any increase in its time complexity. Thus, unless stated otherwise, all our algorithms are designed for $\zeroncc$.
\end{remark}

\section{Preliminaries}
\label{s:prelim}

Sections \ref{sec:structural} and \ref{sec:comp} describe several fundamental \emph{structural} and \emph{computational primitives} in $\zeroncc$ that are used quite extensively later on. Then, Section~\ref{sec:shh} briefly describes the classical (sequential) Havel-Hakimi algorithm~\cite{H55,Ha62}, 
which serves as a basis for many of our algorithms. 
\subsection{Structural Primitives} \label{sec:structural}
Structural primitives deal with arranging the nodes in suitable ways. Specifically, we show how to connect the nodes in the form of a tree, or linearly, sorted by some parameter of the nodes. 
Recall that the initial knowledge graph $G_k$ in our model is a directed \emph{path}. In one round, the directed path can be converted into an undirected (but ordered) path, by having $u$ send its \ID to $v$, for each edge $(u,v)$ in $G_k$.
We say that the path is ordered because, for each $(u,v)$ in the initial knowledge graph, $u$ can remember $v$ as its successor and $v$ can remember $u$ as its predecessor.

\subsubsection{Balanced Binary Search Tree}
Our first goal is to rearrange the nodes in the form of a balanced binary search tree of height  at most $O(\log n)$. The formed tree has to be a search tree, in the sense that, for every node $u$ in the tree, the nodes in the left subtree must appear earlier (i.e., precede $u$)  in $G_k$ and the nodes in the right subtree must appear later (i.e., succeed $u$) in $G_k$. This enables as to find nodes via the knowledge of their position in the path $G_k$.

\label{subsubsec:balancedbinarytree} As a warm up, we first present a simple and straightforward approach to building a balanced binary tree that will {\em not} be a search tree. To do that,
we exploit the fact that the nodes are arranged in a path, and that in one round, each node can learn the ID of its neighbors' neighbors (when present). This can be accomplished by having every node send each of its neighbors' addresses to the other neighbor. With this ability to learn one's neighbors' neighbor, we can decompose a path into two paths comprising nodes in odd and even positions, respectively. Now our algorithm can be described succinctly in the  following recursive manner. Initially, there is a single path, but as the algorithm progresses, several paths will be created and the recursive step, described as follows,  must be applied to every path.  The left-most node $r$ makes its immediate neighbor $a$ its left child; $r$ then makes $a$'s other neighbor $b$ its right child. Then, $r$ removes itself from the path and the path is decomposed into two paths comprising nodes in odd and even positions, with $a$ and $b$ being the  left-most nodes in the two paths respectively. This process is then repeated recursively (and in parallel) in the two paths that are created, and the recursion terminates when empty paths are created. In $O(\log n)$ rounds, all recursive calls would have terminated because the length of the paths halve in each recursive call.  

\begin{center}
  \includegraphics[height=65mm, width=.55\textwidth]{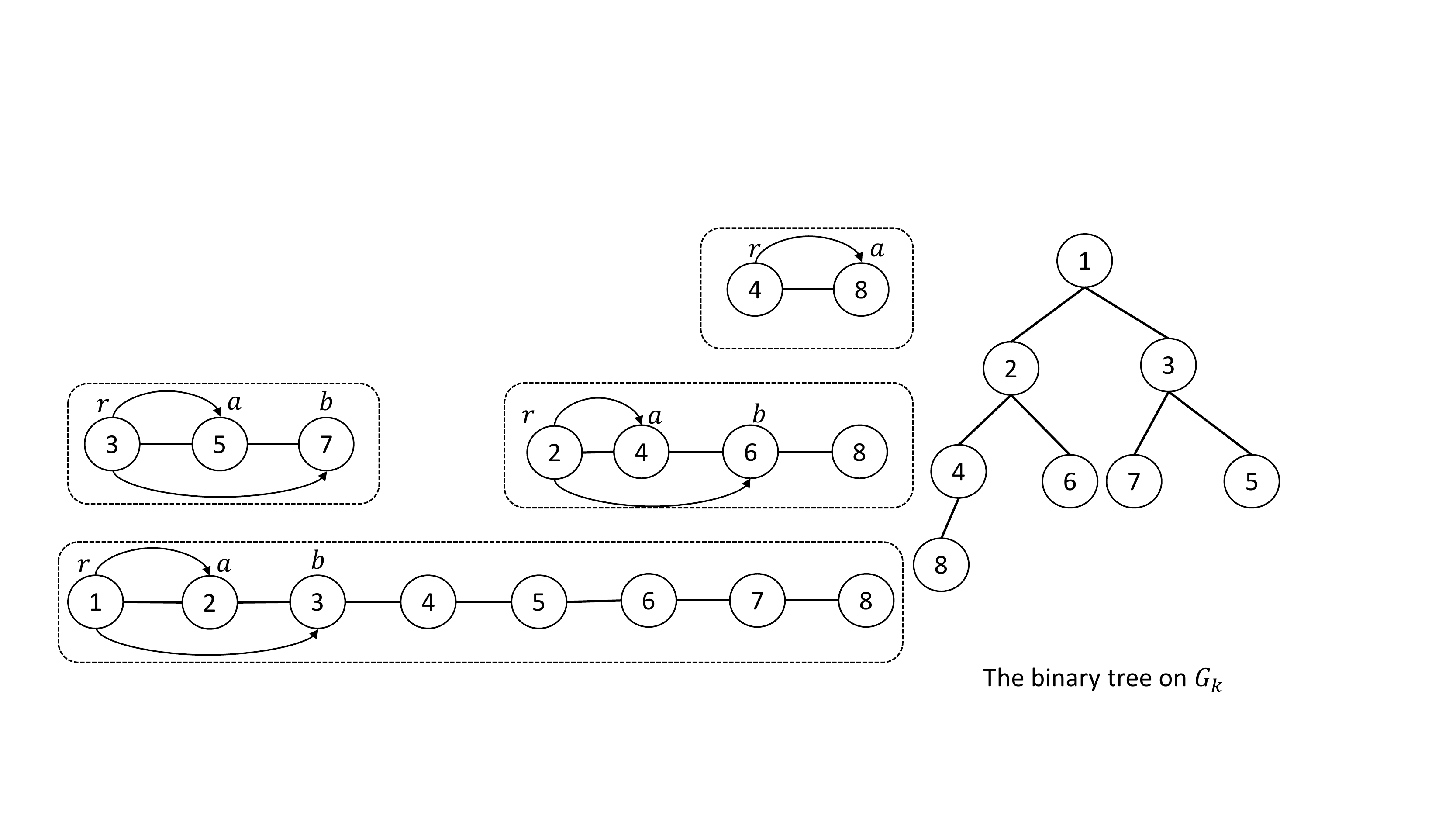}
  \captionof{figure}{The balanced binary tree built on $G_k$ using the recursive construction. For any path, $r$, $a$, and $b$ indicates the parent and its left and right child respectively.}
   \label{fig:btree}
\end{center}

We now provide a more elaborate algorithm to produce a  balanced binary search tree. We first design a structure $\L$ and then perform a controlled breadth first search that results in a balanced binary tree. Our construction is fully distributed and deterministic.
The structure $\L$ comprises $\lceil \log n \rceil+1$ levels $(L_0, L_1, \ldots, L_{\lceil \log n \rceil})$. Each level $i$, $0 \le i \le \lceil \log n \rceil$, comprises up to $2^i$ paths (set of nodes with structure similar to $G_k$, but undirected). $L_0$ has just one path, which is the undirected form of $G_k$. Subsequently, each path $\ell$ on level $L_{i-1}$ produces up to two paths $\ell_0$ and $\ell_1$ on level $L_i$; $\ell_0$ comprises nodes at even position in $\ell$ and $\ell_1$ contains the nodes in odd position. We say that path $\ell$ is the {\em parent path} of paths $\ell_0$ and $\ell_1$, and conversely, $\ell_0$ and $\ell_1$ are the {\em children} of $\ell$. In order to bypass the problem of nodes having to determine the parity of their positions, we do not require nodes in $\ell_0$ and $\ell_1$ to be aware of whether they are in $\ell_0$ or $\ell_1$. 
To construct $\ell_0$ and $\ell_1$, each node in $\ell$ sends its predecessor's \ID to its successor and vice versa and then each node $v$ in $\ell$ connects to its grand-predecessor $p^*$ (i.e., predecessor's predecessor) in $\ell$ and (the similarly defined) grand-successor $s^*$ (both, only if present). Node $p^*$ and $s^*$ now serve as the predecessor and successor, respectively, of $v$ in the newly formed path on level $L_i$. After all paths on level $L_{\lceil \log n \rceil}$ are constructed, the first node in $G_k$ (i.e., the node with no predecessor on level $L_0$) becomes the root node $r$ and initiates a controlled BFS as described in Algorithm~\ref{alg:bfs} to form a BFS tree $T_{BST}$. 

\begin{theorem}\label{thm:binarytree}
The fully distributed deterministic construction outlined above takes a linear arrangement of nodes as the knowledge graph $G_k$ and produces a  binary tree $T_{BST}$ of height at most $\lceil \log n \rceil + 1$ in $O(\log n)$ rounds. Moreover, an inorder traversal of this tree yields the original $G_k$. 
\end{theorem}

\begin{center}
  \includegraphics[height=65mm,width=.55\textwidth]{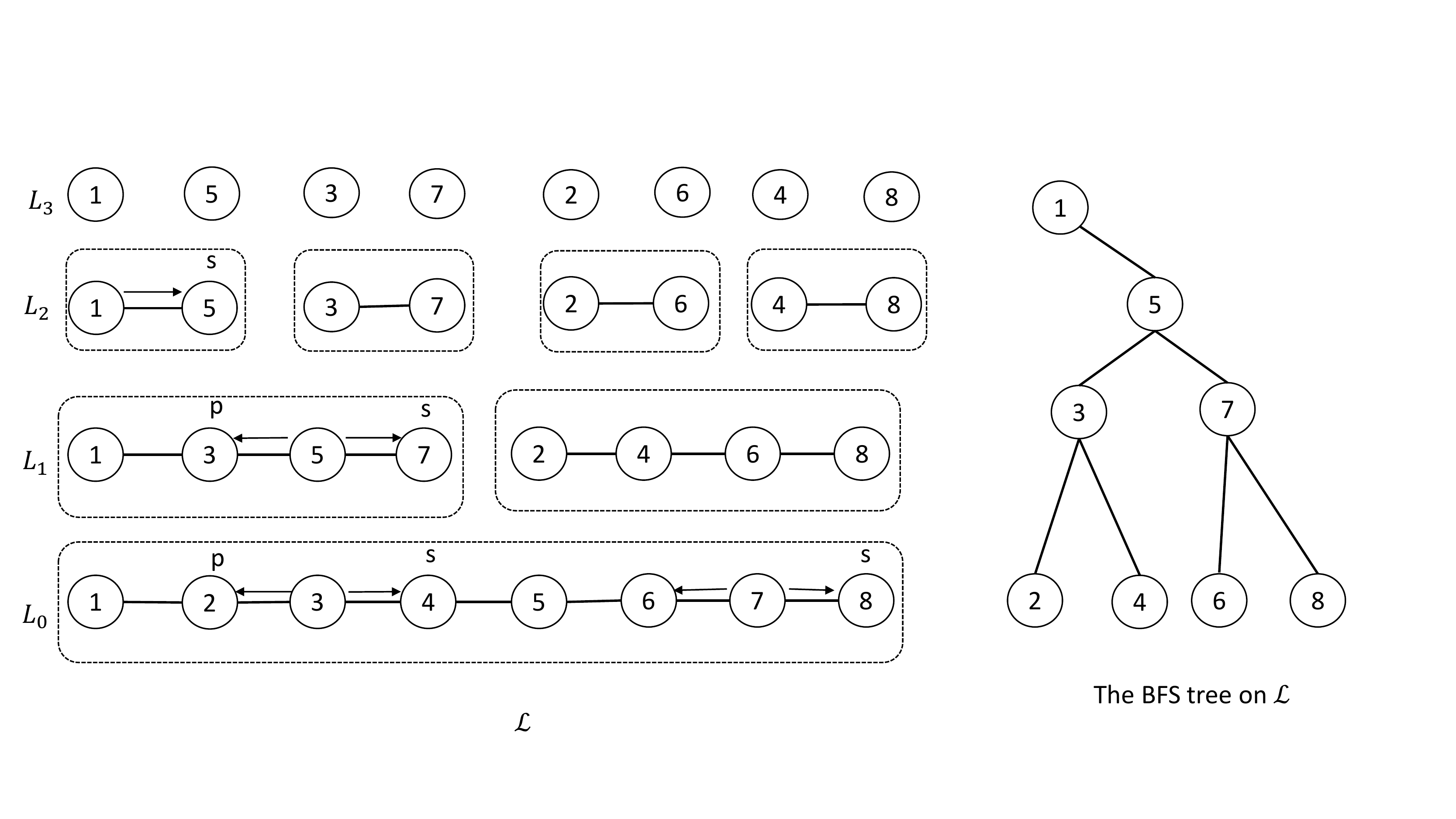}
  \captionof{figure}{The structure $\L$ and the corresponding balanced binary search tree built on $\L$. The letters p and s indicate the \emph{predecessor} and \emph{successor} of a node in a level respectively.}
   \label{fig:bstree}
\end{center}

\begin{algorithm}[!ht]\label{alg:bfs}
 \KwIn{$\L$ has been constructed and the root $r$ knows that it's the root.}
 \KwOut{A balanced binary tree rooted at $r$ such that an in-order traversal will yield the nodes in the input knowledge graph ordering.}
\tcc{Sets $S_p, S_s$ are maintained implicitly in the sense that each node in $S_p$ or $S_s$ knows its membership.}
    $S_p \leftarrow S_s  \leftarrow \{r\}$\;
 \For{$i$ from $\lceil \log n \rceil -1$ down to 0}{
        \ForEach{$v \in S_p$ (in parallel)}{
            \If{$v$ has a predecessor $p$ on level $i$}{
                $v$ invites $p$ as its left child in the tree\;
                $S_p \leftarrow S_p \setminus \{v\}$\; 
            }
    	}
    	\ForEach{$v \in S_s$ (in parallel)}{
            \If{$v$ has a successor $s$ on level $i$}{
                $v$ invites $s$ as its right child in the tree\;
                $S_s \leftarrow S_s \setminus \{v\}$\; 
            }
    	}
        \ForEach{node $u$ not already in the BFS tree that was invited (in parallel)}{
            $u$ accepts the invitation from, say, $v$\;
            $u$ connects to $v$ forming a BFS tree edge between parent $v$ and child $u$\;
            $S_p \leftarrow S_p \cup \{u\}$\;
            $S_s \leftarrow S_s \cup \{u\}$\;
        }
    }
 \caption{A controlled BFS algorithm to arrange nodes in the form of a balanced binary tree.}
\end{algorithm}

\begin{proof}
The paths on each level are close to half the length of their parent path on the previous level, so there are clearly at most $O(\log n)$ levels and the construction of each path takes at most $O(1)$ parallel rounds. Thus $\L$ is constructed in $O(\log n)$ rounds. Since, in the binary tree construction, invitations and subsequent acceptances happen in parallel, each iteration only takes $O(1)$ rounds  of height at most $O(\log n)$, thereby making the overall running time $O(\log n)$.
To argue the correctness of the binary tree construction, we observe the following properties of $T_{BST}$. 

\dnsparagraph{Single Parent.~}
Each node has exactly one parent because it only accepts one invitation. (Once it joins $T_{BST}$, it stops accepting invitations.)

\dnsparagraph{Two Children.~}
Each node $v$ has at most two children, since it gets at most one predecessor child and one successor child. 

\dnsparagraph{Balanced.~}
The depth of $T_{BST}$ is at most $O(\log n)$ (which again holds because there are only $O(\log n)$ levels in $\L$ and each parent is always at a higher level than its children).

\dnsparagraph{Spanning Tree.~}
Finally, we argue that $T_{BST}$ spans all the nodes in $V$. Observe that  the $2^i$ lists on level $i$ form a partition of $V$ because each list on level $i > 0$ is formed as a partition of a list on level $i-1$. We use $t_i$ to denote the list on level $i$ that contains the root node $r$. Thus, $t_{\lceil \log n \rceil}$ is just the singleton list containing  $r$ and $t_0$ is the initial knowledge graph which is a linear arrangement of all nodes in $V$. For each node $v$, define its elevation $\elev(v) \triangleq \max\{i \mid v \in t_i\}$. To show that $T_{BST}$ spans $V$, we need to argue that each node $v$ is added to $T_{BST}$ on level $\elev(v)$. Suppose not. Then, there must be a $v$ with maximum elevation that did not get an invitation on level $\elev(v)$ (and of course it couldn't have gotten an invitation at any  level higher than $\elev(v)$ because $v \not\in t_i$ for $i> \elev(v)$). 

Consider the case where $v$ has exactly one neighbor $v'$ on level $\elev(v)$; Wlog, let $v$ be the successor of $v'$ on level $\elev(v)$. The node $v'$ must have a higher elevation as both $v$ and $v'$ cannot have the same parity in $t_{\elev(v)}$.  Thus, by our assumption, $v'$ must have been invited, entered $T_{BST}$, and put itself in both $S_p$ and $S_s$ on level $\elev(v')$. From level $\elev(v')$ down to $\elev(v)+1$, $v'$ could not have had any successor, as otherwise, $v$ would not be its successor on level $\elev(v)$, so $v'$ must invite $v$ on level $\elev(v)$, contradicting our assumption.  

The case when $v'$ has two neighbors in $t_{\elev(v)}$ can be argued similarly. In fact, we can argue that exactly one neighbor $v'$ of $v$ in $t_{\elev(v)}$ will have elevation $\elev(v)+1$. To see this, notice that both neighbors will be in  $t_{\elev(v)+1}$ as they had the same parity in $t_{\elev(v)}$ and only one can be on level $t_{\elev(v)+2}$ as they both have different parities in $t_{\elev(v)+1}$. Thus, since $v'$ entered $T_{BST}$ on level $\elev(v)+1$, it must have invited $v$ on level $\elev(v)$, which again contradicts our assumption. 

Thus, $T_{BST}$ is indeed a binary tree containing all $n$ nodes and  is balanced. Moreover, for every node $v$, every node in its left (resp., right) subtree appears before (resp., after) $v$ in the original path. Therefore, the original ordering can be recovered by an in-order traversal of $T_{BST}$.
\end{proof}

Finally, we show that each node can compute its position in the path graph $G_k$ in $O(\log n)$ rounds. We first construct $T_{BST}$ on $G_k$. Recall that the inorder traversal of $T_{BST}$ is $G_k$, so the problem reduces to  each node calculating its inorder traversal number in $T_{BST}$. This can be done efficiently in $O(\log n)$ rounds~\citep{AS18}. The key idea is to perform a bottom up phase in which nodes learn the size of their (and their children's) subtrees, and then a top down phase in which the nodes learn their inorder traversals. We assume that $n$ is known, but the algorithm can be modified to work as long as we have a reasonable upper bound for it. In such situations, the exact number of nodes can inferred from $T_{BST}$ as the inorder traversal number of the last node in the original ordering. Thus, once  $n$ becomes common knowledge, the median node's ID can also be inferred in $O(\log n)$ rounds (and flooded to all nodes if needed). Thus, we have
\begin{corollary}\label{cor:inorder}
The position of a node in a path graph can be found in $O(\log n)$ rounds. Similarly, the address of the median node can be made common knowledge to all nodes within $O(\log n)$ rounds. 
\end{corollary}

\subsubsection{Sorting in $\zeroncc$} \label{sec: sorting}
Consider the $\zeroncc$ model with knowledge graph $G_k$ being a path of $n$ nodes. Furthermore, let each node be assigned a unique value. We now show how to build a sorted path on those $n$ nodes. To be more precise, we define a {\em sorted path} as a path of nodes such that the first node, called the {\em head}, has the smallest value, the last node, called the {\em tail}, has the largest value, and all other nodes have values greater than (resp., smaller than) their predecessor (resp., successor). Every node in a sorted path must know its predecessor and successor's IDs (when present). In addition, for convenience, we assume that a balanced binary search tree has been constructed on the set of nodes, i.e., a tree of height that is asymptotically at most logarithmic in the number of nodes such that for every node $v$, the values in its left (resp., right) subtree are smaller (resp., greater) than $v$'s value. The head node's ID will serve as the handle for the sorted path.
 
Our approach is to first build a balanced binary search tree $T$ (see Theorem~\ref{thm:binarytree}) that may not be ordered according to node values, and then build a sorted path of the $n$ nodes in a bottom up manner starting with the leaves. To build the sorted path, first each leaf node sends its value up to  its parent. Then the parent sorts them, constructs the sorted path with itself and its children, and informs the leaves of their neighbors in the sorted path. This is then repeated at higher levels, i.e., each node $u$ passes on the ID of the head node  of the sorted path containing all nodes in $u$'s subtree up to its parent $v$. Node $v$ will receive up to two such IDs, one from each of its children. Node $v$ then uses the merging techniques described shortly  to merge the two corresponding paths (also including itself) and passes on the merged path's head ID to its parent. This bottom-up procedure will therefore require us to go through $O(\log n)$ levels. At each level, the nodes must merge the two sorted paths they receive from their children, for which we  provide an $O(\log^2 n)$ rounds procedure. Thus, the overall round complexity is $O(\log^3 n)$.

We now focus on a node $v$ that has received two head node IDs corresponding to the sorted paths $R_1$ and $R_2$ of nodes from $v$'s left subtree and right subtree, respectively. Node $v$ initiates a recursive merge procedure that combines the two paths into one sorted path $R$  and then inserts itself into the sorted path (taking advantage of the fact that the sorted path $R$ will include a balanced binary search tree). Thus, the output will be a sorted path containing all nodes in the subtree rooted at $v$. 

The merge procedure is defined recursively. Without loss of generality, let $R_1$ be the larger of the two paths (breaking ties arbitrarily) that $v$ receives from its children in $T$. If $R_2$ has just one node $y$ in it, that singleton node $y$ and $v$ are just inserted into $R_1$ to produce the required merged sorted path $R$. Otherwise, The median node $x$ of $R_1$ is computed (see Corollary~\ref{cor:inorder}) and each of the two sorted paths is split into two sorted paths as follows. Path $R_1$ is split into $R_1^<$ (resp., $R_1^>$) comprising nodes in $R_1$ with values less than (resp., more than) the value of $x$. Path $R_2^<$ and $R_2^>$ are defined similarly. These paths can be constructed in a straightforward manner by searching for the value of $x$ in the two paths $R_1$ and $R_2$. For example, the node in $R_1$ with the largest (resp., smallest) value that is less than (resp., more than) the value of $x$ will be $R_1^<$'s tail (resp., $R_1^>$'s head) and having computed the required tails and heads, we can apply Theorem~\ref{thm:binarytree} to complete the construction of the required sorted paths. Thus, these sorted path constructions will take $O(\log n)$ time.  We now recursively merge $R_1^<$ and $R_2^<$ to form the merged path $R^<$ comprising all nodes from $R_1$ and $R_2$ whose values are less than that of $x$. The path $R^>$ whose values are greater than that of $x$ is also obtained similarly. The two paths $R^<$ and $R^>$ are sorted, so they can now be merged easily (with $x$ placed in between) to form the required sorted path $R$. 

The correctness of the procedure is straightforward and similar to mergesort. Note that the recursion depth of each merge procedure will be at most $O(\log n)$ because the sum of the nodes in the two paths at any recursive step will be at most 3/4 of the number of nodes in the two paths at the previous recursive step. This follows from our choice of $x$ as the median element in the larger of the two paths. Thus, the depth of the recursion employed by the merge procedure is $O(\log n)$. At each recursive step, it is required to compute the median element (see Corollary~\ref{cor:inorder}) of the larger sorted path, search within the sorted paths, and build a balanced binary search tree (see Theorem~\ref{thm:binarytree}), each requiring at most $O(\log n)$ rounds. 

\begin{algorithm}[!ht]
\BlankLine
\KwIn{A node  $v$ in the network receives two paths $R_1$ and $R_2$.}
\KwOut{The path $R$, where $R$ is created by merging (in a sorted order) of $R_1$ and $R_2$. }
\BlankLine 
Create balanced binary search trees for both $R_1$ and $R_2$ (see Corollary~\ref{cor:inorder}).

\uIf{either $R_1$ or $R_2$ has only one node $y$.}
{Insert $y$ in the appropriate position using the balanced  binary search  tree\;
Return $R$.}
\Else{
Find the median of the larger path using the balanced  binary trees\;
Use it to split $R_1$ and $R_2$ into four sub-paths $R_1^<, R_1^>$ and $R_2^<, R_2^>$\;
 \tcc{Here $R_1^<$ (respectively $R_1^>$) represents the parts of path $R_1$ whose values are less than (respectively, greater than) that of the median. $R_2^<, R_2^>$ are defined similarly.}
 Call the recursive merge procedure on the pairs $R_1^<, R_2^<$ and $R_2^>, R_2^>$.}

\caption{Recursive-Merge}
\label{Algorithm:rec-merge}
\end{algorithm}

\begin{theorem}
There exists an algorithm for creating a sorted path graph on $n$ nodes of a balanced tree $T$ in $O(\log^3 n)$ rounds in the $\zeroncc$ (and hence also in $\onencc$) model. 
\end{theorem}

\subsection{Computational Primitives}
\label{sec:comp}
The computational primitives described next are  primarily for aggregating, collecting, broadcasting and multicasting information in $\zeroncc$. 
The following 
problem variations have been presented in~\cite{AGGHSKL19} for the $\onencc$ model. We briefly state how we adapt them to $\zeroncc$ and restate the results for the sake of completeness.

\subsubsection{Global Computational Primitives} 
In {\em global broadcast}, a token $\tok_0$ initially held by a designated leader node $\ell$ is to be sent to all other nodes. 
After constructing the balanced binary tree, the leader can simply send the token to the root and the root can then send it down to all nodes.

To introduce {\em global aggregation}, we first need to define a few terms. An {\em aggregate function} is a function on any given set of items $I$. In the context of distributed computing, we may encounter the items in $I$ in various orders, so we use a special form of aggregate functions called {\em distributive aggregate functions}. An aggregate function $f$ is distributive if there is another aggregate function $g$ such that for every partition of $I$ into sets $I_1, I_2, \ldots, I_k$, $f(I) = g(\{f(I_1), f(I_2), \ldots, f(I_k)\})$. We assume that $g$ is known whenever $f$ is known; in fact, they are often the same (e.g, maximum and minimum).

In global aggregation, a designated leader node $\ell$ is known to all nodes, and every node $u$ in $V$ has an input, $\val_u$, of size $O(\log n)$, and is aware of a distributive aggregate function $f$ over $I$ that produces $f(I)$ of size at most $O(\log n)$. The goal is for the leader to learn $f(\{\val_v\}_{v \in V})$. Again, constructing a balanced binary tree, we may aggregate the values to the root using the standard convergecast technique, and the root will then send $f(I)$ to $\ell$.

\begin{theorem}\label{thm:global}
Global broadcast and global aggregation can be performed in $O(\log n)$ rounds.
\end{theorem}

\subsubsection{Global Collection} Again, a designated leader $\ell$ is known to all nodes. For some $A \subseteq V$ of size $|A| = k$, each $v \in A$ has a token $\tok_v$ of size $O(\log n)$ bits and the goal is to collect $\{ \tok_v \}_{v \in A}$ at $\ell$.  Using the balanced binary tree and pipelining, we get:
\begin{theorem}{\label{thm:global_collection}}
Global collection can be performed in $O(k+ \log n)$ rounds. 
\end{theorem}

\subsubsection{Local Computational Primitives} The first local task is {\em local aggregation} in which we seek to aggregate data over  $g$ different aggregation groups $A_1, A_2, \ldots, A_g \subseteq V$. 
Each node $v \in A_i$ has an input, $\val_{v,i}$. The $A_i$'s are not necessarily disjoint, so the number of input values each $v$ has equals the number of $A_i$'s that contain it. Each aggregation group $A_i$ has an associated destination node $t_i$ (not necessarily in $A_i$) that must learn $f_i(\{\val_v\}_{v \in A_i})$ for some distributive aggregate function $f_i$ known to all $v \in A_i$. Each aggregation group $A_i$ has a unique identifier\footnote{In~\cite{AGGHSKL19}, the unique group identifier happens to be the \ID of $t_i$, but as commented in~\cite{AGGHSKL19}, any arbitrary unique identifier will suffice equally well. Some of our algorithms, however, crucially require choice of group identifiers other than the \ID of $t_i$.}
$\gid_i \in [n]$ known to all its members. 

Define $L = \sum_i |A_i|$, $\ell_1 = \max_{v \in V}|\{i \in [g] \mid v \in A_i\}|$, and $\ell_2 = \max_{v \in V}|\{i \in [g] \mid v = t_i\}|$. 

The participants of $A_i \cup \{t_i\}$ are not required to know the other participants. 

To perform the aggregation task (and other subsequently defined tasks), we need to ensure that the nodes in $V$ can emulate a butterfly network (see~\cite{MU17} for definition and details) within the framework of $\zeroncc$. This can be accomplished in $O(\log n)$ time by adapting a recursive procedure given in~\cite{AS18};  also see~\cite{AGGHSKL19} for specifications of the required emulation.

\begin{theorem}[\cite{AGGHSKL19}, Thm. 2.2]
\label{thm:aggregation}
Aggregation can be performed by a randomized algorithm in 
$O\left (\frac{L}{n}+ \frac{\ell_1 + \ell_2}{\log n}+\log n\right )$
rounds w.h.p. 
\end{theorem}
The second primitive is {\em local multicasting}. This task concerns sets $A_1, A_2, \ldots, A_g \subseteq V$ and source nodes $\{s_1, s_2, \ldots, s_g\}$. Each {\em multicast group} $A_i\cup\{s_i\}$ has a unique group identifier $\gid_i$, known to all its participants. Each source node $s_i$ has a token $\tok_i$ that must reach all nodes in $A_i$. The parameters $L$ and $\ell_1$ are  defined as before. Also let $\ell_3 = \max_{v \in V}|\{i \in [g] \mid v = s_i\}|$. 
\begin{theorem}[Theorem 2.5 in~\cite{AGGHSKL19}]\label{thm:multicast}
Multicast can be performed in 
$O\left(\frac{L}{n} + \frac{\ell_1 + \ell_3}{\log n} + \log n\right)$
rounds.
\end{theorem}

Finally, we consider the task of collecting tokens, where we again have groups $A_1, A_2, \ldots, A_g \subseteq V$, and each $v \in A_i$ has an input token $\tok_{v,i}$ of size at most $O(\log n)$. the $A_i$'s are not necessarily disjoint, so the number of tokens at each $v$ equals the number of $A_i$'s that contain $v$ (assuming $v$ has a different token for each group it belongs to). Each group $A_i$ has an associated destination node $t_i$ (not necessarily in $A_i$) that must collect the tokens in $\{\tok_{v,i}|v \in A_i\}$. As usual, each group $i$ has a unique group \ID associated with it.
For simplicity, we assume that no two groups share the same destination. 
\begin{theorem}\label{thm:token}
Token collection can be performed in $O\left (\frac{L}{n} + \frac{\ell_1
}{\log n} + \log n \right)$ rounds. 
\end{theorem}

\begin{proof}
The goal is to map the given token collection task into a suitable aggregation task. Our approach is to create aggregation groups -- one group for each token -- that comprises just two nodes each, namely, the source of the token and its destination. Concretely, consider $\tok_{v,i}$ at $v \in A_i$ that must reach $t_i$. Then, we need to form an aggregation group with the singleton set $\{v\}$ and destination $t_i$. The challenge is to get both nodes in each group to agree on a group \ID because, if we use any property of the source $v$ (either its \ID or position), the destination $t_i$ may not know it. Instead, we form a unique group \ID by concatenating the following: the group \ID of group $i$ (in the token collection instance) and $u_i(v)$, where $u_i(v)$ is the inorder traversal number of  $v$ in the multicast tree~\cite{AGGHSKL19} that is used by group $i$, which we know how to compute from Corollary~\ref{cor:inorder}. Now we can apply Theorem~\ref{thm:aggregation} to get the required result.
\end{proof}

\subsection{Sequential Havel-Hakimi Algorithm}\label{sec:shh}
The characterization of Havel~\cite{H55} and Hakimi~\cite{Ha62} for graphic sequence can be stated concisely as follows. 

\begin{theorem}[Based on~\cite{H55} and \cite{Ha62}] A non-increasing sequence $D=(d_1,d_2,...,d_n)$ is graphic if and only if the sequence $(d'_2,...,d'_n)$ is graphic,
where $d'_j =d_j -1$, for $j\in [2,d_1 +1]$, and $d'_j =d_j$, for $j\in [d_1 +2,n]$
\end{theorem}

\begin{center}
  \includegraphics[height=60mm, width=.49\textwidth]{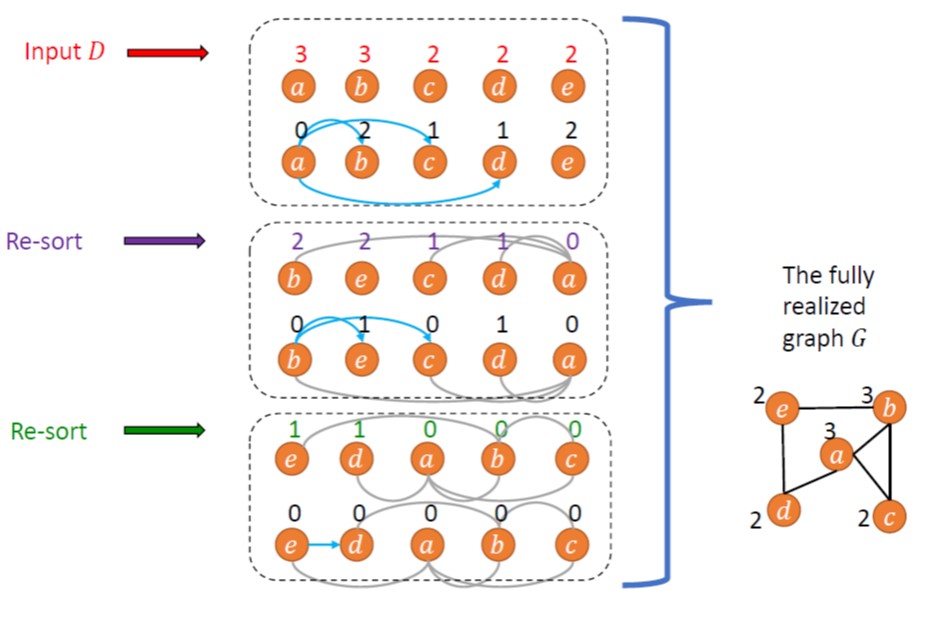}
  \captionof{figure}{A demonstration of the sequential Havel Hakimi. In each step, in the sorted degree sequence (sorted in non-increasing order of degrees), we realize the current highest degree.}
   \label{fig:havelhakimi}
\end{center}

This characterization directly implies an $O(\sum_{i=1}d_i)$ time sequential algorithm, known as the Havel-Hakimi algorithm, for constructing a realizing graph $G=(V,E)$  where $V=\{v_1,...v_n\}$ and $deg(v_i)=d_i$, 
or deciding that no such graph exists. The algorithm works as follows.
Initialize $G=(V,E)$ to be an empty graph on $V$.
In step $i$: (a) remove $d_i$ from $D$, and set $d_j=d_j-1$ for all $j \in[i+1,d_i+i+1]$; (b) set the neighborhood of the node $v_i$ to be $\{v_{i+1},v_{i+2},...v_{i+1+d_i}\}$;
(c) finally, sort the updated $D$ (as well as associate nodes). 
If, at any step, $D$ contains a negative entry, then the sequence is not realizable. See Figure~\ref{fig:havelhakimi} for an illustration.

\section{\hbox{Distributed Degree Realization in Graphs}}
\label{sec:deg-realiz-general}

In Subsection \ref{ss:implicit} we present an $\tO(\min\{\sqrt{m},\Delta\})$ time algorithm for implicit realization, where $\Delta$ is the maximum degree and $m$ is the number of edges in the realizing graph. In Subsection \ref{ss:explicit} we extend this to an $\tO(\Delta)$ time algorithm for explicit realization. Then, in Subsection \ref{ss:approx} we present a distributed degree realization algorithm that, allowing for small variations in the assigned degrees, obtains a realization in the scenarios where the degree sequence is not graphic. This plays a crucial role in the connectivity threshold realizations presented later in Section~\ref{section:Connectivity}.


\subsection{Implicit Degree Realization in \texorpdfstring{$\tO(\min\{\sqrt{m},\Delta\})$}~ Time}
\label{ss:implicit}

We first give a high-level description of our approach to implicit  degree realization. 
It is a parallel version of the well known
\hh procedure~\citep{Ha62,H55}. Recall that the \hh algorithm works by  ensuring that the highest degree node is satisfied first (and once satisfied, its degree is set to $0$). In our case, we slightly adapt the algorithm to allow several nodes of the highest degree to be satisfied in parallel. We ensure that in every round, the nodes having the highest degree disappear (they are assigned their neighbors in the final realizing graph $G$). 

 Our algorithm proceeds in phases. Each phase operates as follows. We first sort the nodes in non-increasing order of degrees in $O(\log^3{n})$ rounds via techniques presented in Section \ref{s:prelim}. These nodes are arranged in a path graph $L$. Each node determines its position in $L$ (see Subsection~\ref{sec:structural}). Let $\delta_j$ be the value of the maximum degree of any node in the $j^{th}$ phase of the algorithm and let $N$ denote the number of nodes of degree $\delta_j$ in that phase. 
 Define 
 $$q ~=~ \max\left\{1,\left\lfloor \frac{N}{\delta_j+1}\right\rfloor\right\}~.$$ 
 Then on $L$, starting with the highest degree nodes, divide the first $q(\delta_j+1)$ nodes into $q$ multicast groups, $g_1,\ldots,g_q$ of $\delta_j+1$ nodes each.  
 In every group $g_i$, $i \in [q]$, the leftmost node $t_i$ in the path graph $L$ multicasts its \ID to nodes in $g_i$. Thus, every other node $v \in g_i$, knowing $t_i$'s \ID, forms an implicit overlay edge with $t_i$. Note that each group $g_i$ can use $i$ as the group's unique identifier. The nodes $t_i$, $i \in [g]$ set their degrees to $0$, while all other nodes in the groups decrease their degree by $1$. This means that $q$ nodes of degree $\delta_j$ have now disappeared from the graph (as they have been realized and are no longer under consideration). We then re-sort the list and repeat the above procedure until all nodes have degree $0$. The detailed pseudocode is presented in Algorithm~\ref{Algorithm:dist-deg-real}.

\begin{algorithm}[!ht]
\BlankLine
\KwIn{An $n$-node network in which each node $x$ is provided with a degree $d(x)$.}
\KwOut{A corresponding implicit realization that satisfies the degrees assigned to the nodes.}
\BlankLine
$j=0$\\
\While{$(1)$}{
$j \leftarrow j+1$\\
Sort the nodes in the non-increasing order of degrees; refer to the nodes as $x_1,\ldots,x_n$ such that $d(x_1)\geq d(x_2)\ldots \geq d(x_n)$.\\
(After sorting, each node knows its position in $L$.)\\
Broadcast $\delta_j=d(x_1)$, the current maximum degree, to all the nodes in the network.\\
\uIf{$\delta_j\geq 1$}{
	Aggregate and broadcast to all the nodes the value $N=\max\{i~|~d(x_i)=\delta_j\}$.\\
	Let $q=\max\{1,\lfloor \frac{N}{\delta_j+1}\rfloor\}$.\\
	\ForEach{$i\in\{\alpha (\delta_j+1)-\delta_j~|~\alpha\in[q]\}$, in parallel}
	{
		Set $d(x_i)=\textsc{nil}$.\\	
		Broadcast ${\ID}(x_i)$ to the next consecutive $\delta_j$ successors of $x_i$, i.e. $x_{i+1},x_{i+2},\ldots,x_{i+\delta_j}$.\\
		The nodes $x_{i+1},\ldots,x_{i+\delta_j}$ store ${\ID}(x_i)$ in their neighbor-list, and decrease their degrees by $1$.
		In the process, if the degree of a node becomes negative, then it broadcasts $\textsc{unrealizable}$, and the execution terminates.\label{step:unrealizable_versus_reset-degree-to-0}
	}
	}
	\lElse{exit the while loop}
}
\caption{Distributed-Degree-Realization}
\label{Algorithm:dist-deg-real}
\end{algorithm}

\begin{lemma}\label{lem:whileloop}
The while loop in Algorithm~\ref{Algorithm:dist-deg-real} is invoked $O(\min\{\Delta,\sqrt m\})$ times (i.e., there are at most $O(\min\{\Delta,\sqrt m\})$ phases).
\end{lemma}
\begin{proof}
Consider a phase $j$ (comprising the $j^{th}$ iteration of the while loop), where $\delta_j$ is the maximum degree and $N$ is the number of nodes of degree $\delta_j$. 
We distinguish two cases.

If $N\leq \delta_j+1$, then the nodes $x_{2},x_{3},\ldots,x_{1+\delta_j}$ store ${\ID}(x_1)$ in their neighbor-list, and decrease their degree by $1$. Also, $x_1$ is removed by setting its degree to zero. So overall, the maximum degree in graph has decreased by at least $1$. 

Otherwise ($N>\delta_j+1$), for $q=\lfloor N/(\delta_j+1)\rfloor$ distinct values of $i$, the nodes $x_{i+1},x_{i+2},\ldots,x_{i+\delta_j}$ store ${\ID}(x_i)$ in their neighbor-list, and decrease their degree by $1$. Intuitively, we remove $q$ stars from $G$. Though the maximum degree may remain $\delta_j$ after this process, the number of nodes left with degree $\delta_j$ must be at most $\delta_j$. This implies that in the next round, $N$ would be bounded by $\delta_j+1$, so at most one additional round would be needed to eliminate $\delta_j$.

Basically, in each iteration of the while loop, at least one node with the maximum degree for that phase is removed from consideration.
This guarantees the number of iterations to be at most $\Delta$. Observe that the number of nodes of degree greater than $\sqrt m$ in any $m$-edge graph is at most $O(\sqrt m)$. Also, in each iteration, the degree of at least one node of maximum degree becomes zero; it follows that the number of iterations required to remove all nodes with degree $>\sqrt{m}$ is at most $O(\sqrt m)$. For the remaining nodes with degree $\leq \sqrt{m}$, the $O(\Delta)$ bound translates into an $O(\sqrt m)$ bound.
Hence, the number of iterations is bounded by $O(\min\{\Delta,\sqrt m\})$.
\end{proof}

\begin{lemma}\label{lem:phase}
One phase (i.e., a single iteration of the while loop in Algorithm~\ref{Algorithm:dist-deg-real}) requires at most $\tO(1)$ rounds.
\end{lemma}

\begin{proof}

Consider the sorting, aggregation of frequency/ maximum, and selective broadcasting procedures in an individual phase (steps 2-10 in Algorithm~\ref{Algorithm:dist-deg-real}).
The sorting in step 2 is accomplished in $O(\log^3 n)$ rounds via the sorting techniques presented in Sect. \ref{s:prelim}. This allows all nodes to be arranged in the path $L$ (sorted by decreasing order of degrees). Crucially, at the end of step 2 each node knows its position in $L$. We then accomplish step 4 via a global broadcast and inform all the nodes in the network of the value $\delta_j$. Step 6 can then be accomplished by aggregation and broadcast. First, we calculate the value of $N$, by having each node $x_i$ such that $d(x_i)=\delta_j$ aggregate towards the node $x_1$ using the group \ID $1$. Then $x_1$ performs a global broadcast to inform the nodes in $L$ of  $N$. All of the above can be accomplished in $O(\log n)$ rounds using the global aggregation and broadcast techniques discussed in Sect. \ref{s:prelim} (Theorem \ref{thm:global}).

Once $N$ is known to all the nodes in the network, each node can locally compute the value of $q$. In step 7, we form $q$ distinct groups for $i\in[1,q]$. Each $i$ is a group \ID for any node that is going to become a neighbor for $x_i$. Note that a node can calculate which distinct group it belongs to locally (with information about its position and $N$). Thus, for any $i$, each node $x_i$ can use $i$ as the group \ID and use the multicast algorithm to broadcast its \ID to nodes $x_{i+1},x_{i+2},\ldots,x_{i+\delta_j}$. All of this can be done in parallel. Notice that by Theorem~\ref{thm:multicast}, steps 5-7 can be done in $O(\log n)$ rounds.
\end{proof}

We thus obtain the following theorem.

\begin{theorem}
There exists a procedure for implicitly realizing any given length $n$ graphic sequence $D=(d_1,\ldots,d_n)$ in $\tilde{O}(\min\{\sqrt{m},\Delta\})$ rounds, in both the $\zeroncc$ and $\onencc$ models, where
$\Delta$ is the maximum degree in $D$ and $m$ is the number of edges required for the realization.
\label{theorem:deg-seq-gen-1}
\end{theorem}

\begin{proof}
By lemma~\ref{lem:whileloop}, Algorithm~\ref{Algorithm:dist-deg-real} involves at most $O(\min\{\sqrt{m},\Delta\})$ iterations of the while loop. 
By Lemma \ref{lem:phase}, a single iteration of the while loop (steps 2-10) can be performed in $\tilde{O}(1)$ rounds. It follows that the entire construction process takes $\tO(\min\{\sqrt{m},\Delta\})$ rounds.
\end{proof} 

\subsection{An \texorpdfstring{$\tO(\Delta)$}~  Time Algorithm for Explicit Degree-Realization}
\label{ss:explicit}

We now describe the extension of Theorem~\ref{theorem:deg-seq-gen-1} for explicit realization.
We first execute Algorithm~\ref{Algorithm:dist-deg-real}. At the end of its execution, each edge is stored implicitly. That is, for any edge $e=(u,v)$ that was formed, at least one of its endpoints (say $u$) is aware of the edge's existence (and of $v$'s \ID). Therefore, $u$ must communicate its \ID to $v$ to make the realization explicit. 
To accomplish this, we create a group for each $v$ such that the associated set is all nodes that must communicate their \IDs to $v$. 
Now applying Theorem~\ref{thm:token}, we conclude the following.
\begin{theorem}
There exists a procedure for explicitly realizing any given 
length $n$ graphic sequence $D=(d_1,\ldots,d_n)$ in $O(m/n + \Delta/\log n + \log n)$ rounds, in both the $\zeroncc$ and $\onencc$ models,
where $m = (1/2) \sum_i d_i$ and $\Delta = \max_i d_i$. 
\label{theorem:deg-seq-gen-2}
\end{theorem}

\subsection{Approximately Realizing Non-Graphic Sequences}
\label{ss:approx}

We next consider the case of a non-realizable degree sequence $D=(d_1,d_2,\ldots,d_n)$. An {\em upper envelope} to $D$ is a degree sequence $D'=(d'_1,d'_2,\ldots,d'_n)$ satisfying $d'_i\geq d_i$ for every $i$. A natural goal is to find a realization for an upper envelope $D'$ of low total {\em discrepancy} with $D$, defined as $\varepsilon(D,D')=\sum_{i=1}^n (d'_i - d_i)$. The question of minimizing the discrepancy was studied by Hell and Kirkpatrick in the centralized setting~\cite{HellK:09}. We present a distributed solution that provides an explicit realization of an upper envelope $D'$ for a given non-realizable sequence, with a discrepancy of at most $\sum_{i=1}^n d_i$. This only requires the following alteration to step~\ref{step:unrealizable_versus_reset-degree-to-0} of Algorithm~\ref{Algorithm:dist-deg-real}.

\vspace{2mm}
\begin{center}

\fbox{

\parbox{0.46\textwidth}{
  
    \textbf{Step~\ref{step:unrealizable_versus_reset-degree-to-0}:} The nodes $x_{i+1},\ldots,x_{i+\delta_j}$ store ${\ID}(x_i)$ in their neighbor-list, and decrease their degrees by $1$. In the process, if degree of a node becomes negative, then it resets its degree to $0$.
  }
}
\end{center}
\vspace{3mm}

Let $D'=(d'_1,d'_2,\ldots,d'_n)$ be the degree sequence of the output graph. It is easy to see that the  total degree increase $\sum_{i=1}^n (d'_i - d_i)$ is bounded by $\sum_{i=1}^n d_i$. This is because when a node's degree is reset to $0$, the resorting ensures that it will again be used as a neighbor at most $d_i$ times, yielding the following.

\begin{theorem}
There exists a procedure for explicitly realizing, for any given 
length $n$ (possibly non-graphic) sequence $D=(d_1,\ldots,d_n)$, an upper envelope $D'=(d'_1,\ldots,d'_n)$ satisfying (i) $ d'_i\geq d_i$ for every $i$, and (ii)~$\sum_{i=1}^n d'_i\leq 2\sum_{i=1}^n d_i$. 
This can be done in $\tilde{O}(\Delta)$ rounds
in both the $\zeroncc$ and $\onencc$ models, where $\Delta$ is the maximum degree in $D$. 
\label{theorem:deg-seq-gen-3}
\end{theorem}
\section{Degree-Sequence Realization in Trees} \label{s:trees}
In this section, we consider the degree realization problem when we restrict our realizations to trees. We refer to this as the {\em tree-realization} problem. Note that a degree sequence $(d_1, d_2, \ldots, d_n)$ is realizable as a tree if and only if $\sum_i d_i=2(n-1)$ \cite{GraphTheory}. Since this condition can be verified in $\zeroncc$ in $O(\log n)$ rounds by aggregation, we can quickly test whether a given degree sequence has a valid tree-realization. 

We present two $O(\plog{n})$ round algorithms for realizing trees in $\zeroncc$. The first is a simple algorithm yielding a tree-realization of the maximum possible diameter. We then show how with some modification to the first algorithm, it is possible to output a minimum diameter tree.

A detailed pseudocode of our first algorithm is presented in Algorithm~\ref{Algorithm:dist-tree-deg-real1}. The key idea is as follows. 
Given a realizable sequence, we first sort the nodes according to their degrees. This allows us to separate the leaves from the non-leaves. Let $V=\{x_1,x_2,\ldots,x_n\}$ be the nodes in the network such that $d(x_1)\geq d(x_2) \ge \ldots \geq d(x_n)$. Assume that they are arranged in a sorted path graph $L$. Define 
$$k=|\{x~|~d(x)>1\}|.$$ 
Thus, $x_{k+1},\ldots, x_n$ are the leaves. For each $i$ in $[1,k]$, $x_i$ creates an edge with $x_{i+1}$ by exchanging \IDs. Now for any $x_i$ such that $i \in [2,k]$, as that node has a predecessor and a successor, the remaining degree requirement is $d(x_i)-2$; for $x_1$, the degree requirement is $d(x_1)-1$. Each of these can be satisfied by simply attaching the required number of leaves (as the analysis in~\citep{AM1996realizing} shows that there will be a sufficient number of leaves). Once $k$ is made known to the nodes, node $x_i$ is only required to know how many leaves are required for the nodes $x_1,\ldots,x_{i-1}$. In the algorithm, we capture this by the prefix sum $p_i$, 
which allows node $x_i$ to calculate the number of children its predecessors require. Once $p_i$ is known, $x_i$ knows the position of its leaves in $L$ and can inform its \ID to them. 

\begin{algorithm}[!ht]
\BlankLine
\KwIn{An $n$-node network with each node $x$ provided with a degree $d(x)$.}
\KwOut{A corresponding implicit tree realization.}
\BlankLine
Sort the nodes into a list $L$, represented as $x_1,\ldots,x_n$, in non-increasing order of degrees. After sorting, each node knows its position in $L$. \\
Aggregate the value $S=\sum_{x \in V}d(x)$ at a single node (say $x_1$).
If $S\neq 2(n-2)$, then $x_1$ broadcasts $\textsc{unrealizable}$, and the procedure terminates.\\
Aggregate and broadcast to all the nodes the value $k=|\{x~|~d(x)>1\}|$ such that $x_1, \ldots, x_k$ are the non-leaf nodes and $x_{k+1}, \ldots, x_n$ are leaves in $L$.
\\
Compute the prefix sums $p_i=2+\sum\limits_{j=1}^{i-1}\big(d_j-2\big)$ for $2\leq i\leq k$. $p_1=2$.\\

\ForEach{$i\in [k]$, in parallel}
{
    \uIf{$i=1$}{Set $I=0$}
    \lElse{Set $I=1$}
	Node $x_{i}$ stores {\ID}$(x_{i-1})$ and {\ID}$(x_{i+1})$ in its neighbor-list (when exists).\\ 
	Broadcast {\ID}$(x_i)$ to the following $d_i-1-I$ contiguous leaves $x_{k+p_i+I},\ldots, x_{k+p_i+d_i-2}$. \\
	
	The nodes $x_{k+p_i+I},\ldots, x_{k+p_i+d_i-2}$ store {\ID}$(x_i)$ in their neighbor-list. \\ \CommentSty{\%there will always be sufficiently many leaves (From Arikati and Maheswari\cite{AM1996realizing}, section 3.2)} \\
}
\caption{Distributed-Tree-Realization-1}
\label{Algorithm:dist-tree-deg-real1}
\end{algorithm}

We now show how to achieve a {\em polylogarithmic} time implementation of Algorithm~\ref{Algorithm:dist-tree-deg-real1}. Steps 1-3 can be performed by a combination of sorting and aggregation. First, sort the nodes according to their degrees, after which each node can be made to know its position in the path (via the BFS tree construction). Then, perform steps 2 and 3 by global aggregation operations (to the node in position 1, $x_1$).

The prefix sums $p_i$ (in step 4) can be computed in $O(\log n)$ rounds in a manner that is reminiscent of computing inorder traversal numbers. In fact, inorder traversal numbers are (close to) prefix sums when the degree values are 1. Our approach is to build a local binary tree on the nodes $x_1,\ldots x_k$. Then we apply a two-phase process. Phase 1 involves a bottom-up convergecast (to calculate the sum of the degree values of nodes in sub-trees). At the end of phase 1, the root knows its prefix sum and phase 2 involves a recursive top-down computation in the binary tree for nodes to deduce their prefix sums.

 Step 9 can be performed in $O(1)$ rounds by neighbors simply exchanging their \IDs. Note that at the end of step 4, each node $x_i$ (for any $i \in[1,k]$) knows how many leaves it requires. Thus any $x_i$ can locally calculate the positions of its leaves in the path $L$. Suppose $x_i$'s leaves are located in positions $j$ to $j+(d_i-3)$. Then $x_i$ can inform these nodes of its \ID in the following manner. First, using $j$ as the group \ID, $x_i$ informs $x_j$ of its \ID using the aggregation algorithm. Then it does the same for node $x_{j+d_i-3}$. Note that this only requires $O(\log n)$ rounds via aggregation (Theorem~\ref{thm:aggregation}). We now treat the problem of the nodes $x_j$ and $x_{j+d_i-3}$ informing the other nodes between their position in $L$ as a smaller instance of the global broadcast problem in the NCC (containing only the $d_i-2$ nodes), which can be solved in $O(\log n)$ rounds (Theorem~\ref{thm:global}). Also, each $x_i$ can simultaneously inform its \ID to all its leaves without congestion by running the above procedure concurrently. Thus we have the following.

\begin{theorem}
    There exists a procedure for implicitly realizing a length $n$ tree realizable sequence by a tree overlay network in $O(\log^3{n})$ rounds in the $\zeroncc$ model.
\end{theorem}

Next, we present an algorithm for obtaining a minimum  diameter tree realization from a degree sequence.
The resulting Algorithm~\ref{Algorithm:dist-tree-deg-real2} is in fact a distributed version of the sequential algorithm studied in \cite{SSW16}, and the resulting tree is referred to therein as the {\em greedy tree} $T_G$. 

The key idea here is to put the nodes with the higher degrees as high up in the tree as possible, without violating the realization requirement.
Initially, the highest degree node $x_1$ becomes the root and connects to next $d(x_1)$ highest degree nodes. Thereafter, each node $x_i$ satisfies its degree requirement by connecting with $d(x_i)-1$  nodes (as it is already connected to its parent) of highest degree that do not yet have a parent. 

\begin{algorithm}[!ht]
\BlankLine
\KwIn{An $n$-node network with each node $x$ provided its degree $d(x)$ in the tree.}
\KwOut{An 
implicit tree realization of min diameter.}
\BlankLine
Initialization: As in Algorithm~\ref{Algorithm:dist-tree-deg-real1}, line numbers 1 -- 3.\\
Compute the prefix sums $p_i=2+\sum\limits_{j=1}^{i-1}\big(d_j-1\big)$ for $2\leq i\leq k$. Set $p_1=2$.\\
\ForEach{$i\in [n]$, in parallel}
{
    \uIf{$i=1$}{Set $I=0$}
    \lElse{Set $I=1$}
    Broadcast ${\ID}(x_i)$ to the $d_i-1-I$ nodes $x_{p_i+I},\ldots, x_{p_i+d_i-1}$.\\

	The nodes $x_{p_i+I},\ldots, x_{p_i+d_i-1}$ store ${\ID}(x_i)$ in their neighbor-list. \\
}
\caption{Distributed-Tree-Realization-2}
\label{Algorithm:dist-tree-deg-real2}
\end{algorithm}

We next claim that the greedy tree $T_G$ constructed by Algorithm~\ref{Algorithm:dist-tree-deg-real2} has minimum diameter.
\begin{lemma}
\label{lem:min-diam}
    The output of Algorithm~\ref{Algorithm:dist-tree-deg-real2} is a minimum diameter realization of the input degree sequence $D$.
\end{lemma}

\begin{proof}
Let $\cTvd$ be the class of trees realizing $D$. The {\em eccentricity} of a node $v$ in the tree $T$ is $ecc(v,T)=\max_{u\in T}dist_T(v,u)$. 
Let $n_\ell(T)=|\{v \mid ecc(v,T)\leq \ell\}|$.
It is shown in \cite{SSW16} (Theorem 16) that $n_\ell(T_G) \ge n_\ell(T)$ for every non-negative integer $\ell$ and tree $T\in\cTvd$.
Fix $\ell=\delta_G-1$, where $\delta_G$ is the diameter of $T_G$. Then $n_\ell(T_G)\le n-2$. Since $T_G$ maximizes $n_\ell(T)$, also $n_\ell(T')\le n-2$ for every other $T'\in\cTvd$.
Hence, every $T'\in\cTvd$ must have diameter at least $\ell+1=\delta_G$,  establishing the lemma.
\end{proof}

In the NCC model, Steps 1-2 of Algorithm~\ref{Algorithm:dist-tree-deg-real2} can be performed in $polylogarthmic$ rounds via sorting and global aggregation. Step 3 is performed just as step 4 in Algorithm~\ref{Algorithm:dist-tree-deg-real1}. Steps 8-10 can be performed in a similar manner. Once a node $x_i$ knows its childrens' positions in the path (via the prefix sums), it does the following.  Suppose $x_i$'s children are in positions $j$ to $j+d_{i}-2$. Then $x_i$ informs node $x_j$ and $x_{j+d_{i}-2}$ of its \ID (by a multicast). Then these two children inform the other nodes (as before, this can be seen as a smaller instance of broadcast in the NCC with $d_i-1$ nodes). Note that no node belongs to more than one multicast group. Thus, we get the following.

\begin{theorem}
    There exists a procedure for implicitly realizing a length $n$ tree realizable sequence $D$ by an overlay tree network $T$ in $O(\plog{n})$ rounds, in the $\zeroncc$ model, such that $T$ has the minimum diameter possible for the given $D$.
\end{theorem}
While both Algorithm~\ref{Algorithm:dist-tree-deg-real1} and Algorithm~\ref{Algorithm:dist-tree-deg-real2} have similar initial stages (i.e., sorting the nodes into leaves and non-leaves), Algorithm~\ref{Algorithm:dist-tree-deg-real2} guarantees a minimal-diameter tree as illustrated in Figure~\ref{fig:treerealization}.
\begin{figure}[]
    \centering
    \includegraphics[scale=.27, trim=0 100 0 0, clip]{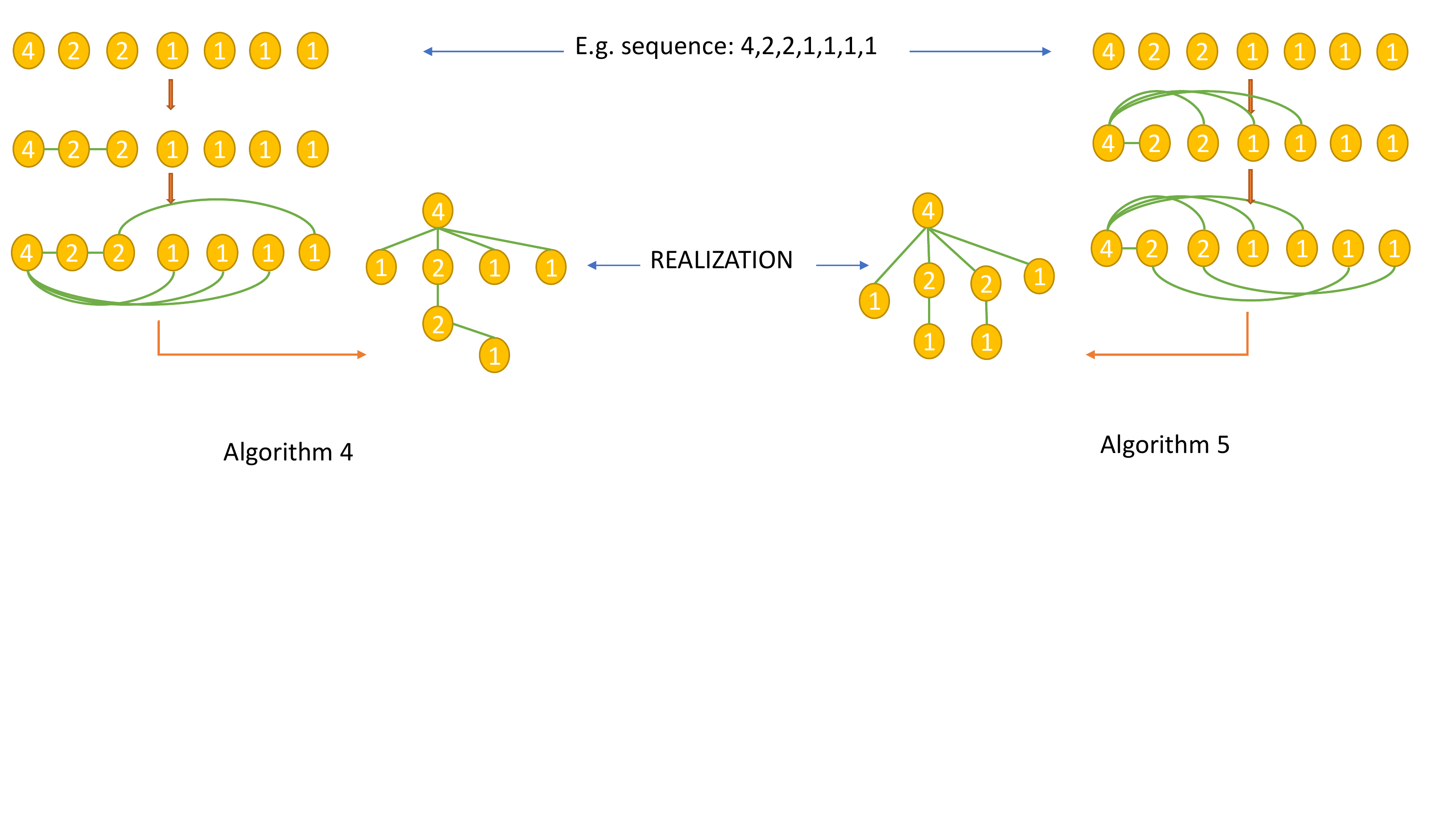}
    \caption{An illustration of our two algorithms for the realization of trees. Given the same degree sequence, clearly Algorithm~\ref{Algorithm:dist-tree-deg-real2} provides a minimal diameter realization.}
    \label{fig:treerealization}
\end{figure}

\section{Connectivity Threshold Realizations}
\label{section:Connectivity}

In this section we study the  minimum connectivity threshold realization problem in the NCC model. For any two nodes $u$ and $v$ in a graph $G$, define their connectivity, denoted as $\edgeconn_G(u,v)$, as the minimum number of edge-disjoint paths present between $u$ and $v$.  In the minimum connectivity threshold problem, each node $v$ in the network is initially provided with a local copy of a {\em connectivity threshold} vector ${\bar\sigma(v)}=\langle \sigma(v,u_1),\cdots,\sigma(v,u_n)\rangle$, which specifies the required minimum edge connectivity 
at $v$ w.r.t. every other node $u\neq v$. The goal is then to compute the sparsest possible graph $G$ 
such that any two nodes $u,v$ in $G$ satisfies the condition
$\edgeconn_G(u,v) \geq \sigma(u,v)$.

We provide a $2$-approximate solution guaranteeing that the number of edges is at most twice the minimum possible. In order to achieve this goal, we do the following.
Let $\rho(v)=\max_{u\in V} \sigma(u,v)$ for each $v\in V$.
Our algorithms output realizations that satisfy the condition
\begin{equation*}
\edgeconn_G(u,v) \geq \min\{\rho(u),\rho(v)\}.
\label{equation:min_conn_strong}
\end{equation*}
Any realization that ensures the above condition guarantees a 2-approximation solution for the connectivity threshold problem. Furthermore, working with the above condition also allows us to assume that each node $v$ can be provided with just the value $\rho(v)$, rather than the entire length $n$ vector ${\bar\sigma}(v)$, or equivalently, that all the entries of ${\bar\sigma}(v)$ are identical to $\rho(v)$.

\subsection{An \texorpdfstring{$\tilde{O}(1)$}~ Time Implicit Realization in $\onencc$}
\label{subsection:conn-1}

We first consider the simplest scenario of implicit realization in the $\onencc$ model. Our algorithm has two steps.
\begin{enumerate}
    \item First, find a node $w$ of maximum $\rho$ value, i.e., such that $\rho(w)=\max_{v\in V}\rho(v)$, breaking ties arbitrarily. Note that $w$ can be found using global aggregation (\ref{thm:aggregation}), and its address can then be broadcast to all the nodes, all in $\tO(1)$ time.
    \item  Next, each node $v\neq w$ chooses an arbitrary subset $X^v=\{x^v_1,\ldots,x^v_{\rho(v)}\}$ of $V$ satisfying that $w\in X^v$, and outputs $X^v\times\{v\}$ as the required number of edges for node $v$. This step can be done in $O(1)$ time in the $\onencc$-model since $v$ already knows the addresses of all the nodes in $X^v$.
\end{enumerate}

\dnsparagraphem{Correctness:~}
Consider any node $v\neq w$. Assume w.l.o.g. that $x^v_1$, the first node in $X^v$, be $w$. Then $(v,w)$, and $(v,x^v_i,w)$ for $i\in[2,\rho(v)]$, are $\rho(v)$ edge disjoint paths from $v$ to $w$, proving that $\edgeconn_G(w,v)=\rho(v)$ for every $v\neq w$. Now for any two nodes $v_1,v_2\neq w$, let $\min\{\rho(v_1),\rho(v_2)\}$ be $\rho(v_1)$. Let $X^{v_1}=\{x_{v_1}^1,\ldots,x_{v_1}^{\rho(v_1)}\}$ and $X^{v_2}=\{x_{v_2}^1,\ldots,x_{v_2}^{\rho(v_2)}\}$ be the neighborhoods of the nodes $v_1$ and $v_2$ respectively (such that $x_{v_1}^1=x_{v_2}^1=w$). For any $i\ne 1$, if $x_{v_1}^i=x_{v_2}^i$ then $(v_1,x_{v_1}^i,v_2)$ is a path from $v_1$ to $v_2$. Otherwise ($x_{v_1}^i\neq x_{v_2}^i$), $(v_1,x_{v_1}^i,w,x_{v_2}^i,v_2)$ is a path from $v_1$ to $v_2$. This ensures that there are $\rho(v_1)$ edge disjoint paths between $v_1$ and $v_2$. 

\dnsparagraphem{Approximation factor:~}
\noindent\emph{Approximation factor}:
Clearly, in any graph that adheres to the minimum connectivity threshold constraints, the degree of each node $v$ must be at least $\rho(v)$. Thus, any such realizing graph must contain at least $\frac{1}{2}\sum_{v\in V}\rho(v)$ edges.
Our algorithm obtains a realization $G$ where each node $v \neq w$ adds exactly $\rho(v)$ edges. Since each node $v\ne w$ adds an edge to $w$, and $w$ itself does not add any edges, the total number of edges is determined by $\sum_{v\in V\setminus\{w\}}\rho(v)\leq \sum_{v\in V}\rho(v)$ edges. Hence, our algorithm gives a 2-approximation solution. We thus get the following.

\begin{theorem}
There exists a $\tO(1)$ time procedure for implicitly realizing, in the $\onencc$ model, 
a given collection of connectivity threshold vectors by a graph $G$ such that number of edges in $G$ is at most twice the number of edges in the optimal realization.
\label{theorem:conn-polylog-time}
\end{theorem}

\subsection{An \texorpdfstring{$\tilde{O}(\Delta)$}~ Time Explicit Realization in $\zeroncc$ and $\onencc$}
\label{subsection:conn-12}

We next present a connectivity threshold algorithm providing an explicit realization that also works  in the $\zeroncc$ model. Our algorithm is inspired by the work of Frank and Chou~\cite{FC70}, which presents a 2-approximation in the centralized setting, and crucially uses the degree-realization results presented in Section \ref{sec:deg-realiz-general}. 

The algorithm begins by first sorting the nodes in non-increasing order of $\rho$ (defined as in Subsection~\ref{subsection:conn-1}). Let $x_1,\ldots,x_n$ be the sorted nodes and let $d_0=\rho(x_1)$.  
The algorithm proceeds in two phases. 
\begin{enumerate}
\item 
The first phase focuses on the nodes $x_1,...,x_{d_0+1}$. The node $x_1$ broadcasts $d_0$ to everyone. Next, the nodes $x_1,...,x_{d_0+1}$ use their assigned $\rho$ values ($(\rho(x_1),\rho(x_2),\ldots,\rho(x_{d_0+1}))$) as a degree sequence to obtain a realization $G_1$ (either exactly or approximately if an exact realization is impossible) using the approximation algorithm given in Section \ref{ss:approx}. 
This ensures that each node $x_i$ in $G_1$, for $i\in[2,d_0+1]$, is connected to $x_1$ and to at least $\rho(x_i)-1$ nodes from the set $x_2,\ldots,x_{d_0+1}$.
\item 
In the second phase, for each $i\in[d_0+2,n]$ in parallel, $x_i$ broadcasts ${\ID}(x_i)$ to $\rho(x_i)$ predecessors of $x_i$, i.e., $x_{i-1},x_{i-2},\ldots,x_{i-\rho(x_i)}$. 
These nodes store ${\ID}(x_i)$ in their neighbor-list. This ensures that for every $i\in[d_0+2,n]$, the degree of $x_i$ in the graph induced by the nodes $x_1,\ldots,x_i$ is at least $\rho(x_i)$. To make the realization explicit, the nodes $x_{i-1},\ldots,x_{i-\rho(x_i)}$ also broadcast their own \IDs to $x_i$ (after receiving $x_i$'s \ID).
\end{enumerate}

\begin{algorithm}[!ht]
\BlankLine
\KwIn{An $n$-node network with each node $x$ provided with a connectivity threshold $\rho(x)$.}
\KwOut{An explicit connectivity threshold realization.}
\BlankLine
Sort the nodes in the non-increasing order of $\rho$, let these be represented as $x_1,\ldots,x_n$. (After sorting, each node knows its position in the sorted list).\\
Broadcast $d_0=\rho(x_1)$ to all the nodes.\\
Obtain a distributed degree-realization for sequence $(\rho(x_1),\rho(x_2),\ldots,\rho(x_{d_0+1}))$ over the first $d_0+1$ nodes using Theorem~\ref{theorem:deg-seq-gen-3}.\\ 
\ForEach{$i\in[d_0+2,n]$, in parallel,}
	{	
		Broadcast ${address}(x_i)$ to $\rho(x_i)$ predecessors of $x_i$,  $x_{i-1},x_{i-2},\ldots,x_{i-\rho(x_i)}$.\\
	
		The nodes $x_{i-1},\ldots,x_{i-\rho(x_i)}$ after receiving the address of $x_i$ store it in their neighbor-list, and also broadcast their own addresses to node $x_i$, which in turn stores them in its neighbor-list.
	}
\caption{Distributed-Connectivity-Realization}
\label{Algorithm:dist-conn-real}
\end{algorithm}

\dnsparagraphem{Correctness:~}
Let $G_1$ be the graph induced by the nodes $x_1,\ldots,x_{d_0+1}$ computed in the first phase, and $G_2=G$ be final graph after the completion of the second phase. We show that $\edgeconn_G(x_1,x_i)\geq \rho(x_i)$, for each $i>1$. 
First consider $G_1$. Since the degree of $x_1$ in $G_1$ is $d_0$, it is adjacent to all nodes in $G_1$. Consider a node $x_i$ for $i\in[2,d_0+1]$. Let $w_1,\ldots,w_{\rho(x_i)-1}$ be $x_i$'s neighbors in $G_1$, other than $x_1$. Then $(x_i,x_1)$ and $(x_i,w_j,x_1)$ for $j\in[1,\rho(x_i)-1]$ are $\rho(x_i)$ edge disjoint paths from $x_i$ to $x_1$. Since $G_1$ is a subgraph of $G$, for each $i \in[2,d_0+1]$, $$\edgeconn_G(x_1,x_i)\geq \edgeconn_{G_1}(x_1,x_i) \geq \rho(x_i).$$
It is easy to prove by induction that for $i\geq d_0+2$, $\edgeconn_G(x_1,x_i) \geq \rho(x_i)$.
Finally, using Menger's Theorem~\cite{goring2000short}, we get that
\begin{eqnarray*}
\edgeconn_G(x_i,x_j) &\geq& \min\{\edgeconn_G(x_i,x_1),\edgeconn_G(x_j,x_1)\}
\\
&\geq& \min\{\rho(x_i),\rho(x_j)\}.
\end{eqnarray*}
 
\dnsparagraphem{Approximation factor:~}
By Theorem~\ref{theorem:deg-seq-gen-3}, the number of edges in $G_1$ is at most $\sum_{i=1}^{d_0+1}\rho(x_i)$. The second phase adds $\rho(x_i)$ edges for every $i\geq d_0+2$. Hence the number of edges in $G$ is at most $\sum_{i=1}^{n}\rho(x_i)$. By arguments similar to those of Subsection~\ref{subsection:conn-1}, $\frac{1}{2}\sum_{i=1}^{n}\rho(x_i)$ is a lower bound on the number of edges in $G$, implying that our algorithm achieves an approximation factor of two.
We thus have the following.

\begin{theorem}
There exists a $\tO(\Delta)$ time procedure for explicitly (as well as implicitly) realizing, in the $\zeroncc$ and $\onencc$ models, a connectivity threshold graph $G$ 
with at most twice the number of edges in the optimal realization.
\label{theorem:conn-delta-time}
\end{theorem}


\section{Lower Bounds for Degree Realization} \label{sec:lb}
In this section, we establish a number of lower bounds applicable to the distributed graph realization problem in $\zeroncc$, some of which are fairly straightforward.
We begin with the simplest case. 

In any instance of the explicit version, a node of maximum degree needs to know the addresses of its $\Delta$ neighbors. This yields the following.
\begin{theorem}\label{theorem:explicit-lower-bound}
Any distributed $\zeroncc$ algorithm for explicit realization of a degree sequence $D$ with maximum degree $\Delta$ requires at least $\Omega(\Delta/\log n)$ rounds for all instances.
\end{theorem}

We now turn our attention to implicit realizations.
Let $\cD_{n,m}$ be the class of length $n$ degree sequences $D=(d_i)_{i \in [n]}$ such that $m=\sum_{i} d_i / 2$, and let $\cD'_\Delta$ be the class of length $n$ degree sequences with maximum degree $\Delta$. 
\begin{theorem}
\label{lem:generallb}
(a) Any distributed algorithm to realize degree sequences in $\cD_{n,m}$ in $\zeroncc$ requires at least $\Omega(\sqrt{m}/\log n)$ rounds. (b) Any distributed algorithm to realize degree sequences in $\cD'_\Delta$ in $\zeroncc$ requires at least $\Omega(\Delta)$ rounds.
\end{theorem}
\begin{proof} 
Let $k=\lfloor\sqrt{m}\rfloor$, and let $\cD^*_{n,m} \subset \cD_{n,m}$ be the family of degree sequences in which $d_i=0$ for all $i>k$.
Consider an arbitrary 
$D^* = (d_1, d_2, \ldots, d_k, d_{k+1}, \ldots, d_n) \in \cD^*_{n,m}$. Any $\zeroncc$ algorithm must ensure that the first $k$ nodes combined learn $\Omega(m)$ \IDs. By the pigeonhole principle, at least one of these nodes
must learn $\Omega(m/k)=\Omega(\sqrt{m})$ \IDs, which requires $\Omega(\sqrt{m}/\log n)$ rounds in $\zeroncc$.

The lower bound of $\Omega(\Delta)$ for realizing degree sequences in $\cD'_\Delta$ can be argued similarly by considering the degree sequence $(d_i = \Delta)_{1 \le i \le n}$.
\end{proof}

\section{Realization with non-preassigned information}\label{sec:non-preassigned}

So far, we have restricted our focus to realizations in which each node $v$ is provided with a certain local constraint (e.g., its degree or its connectivity vector in the overlay graph) that the realization is expected to satisfy. A natural question to ask is how to approach the problem of computing a realization that meets certain requirements defined in terms of a global graph property, rather than local constraints. In this scenario, nodes are not a pre-assigned an individual constraint they need to satisfy; rather, there are certain requirements that are required of the network, that the nodes must work together to satisfy.

Let's focus on the degree-realizability problem in this context. Here, we assume that the degree vector $D=(d_1,d_2,\ldots,d_n)$ is presented in the distributed manner, that is, in the beginning each $v_i$ is provided with the value $d_i$. Our requirement is that the degree sequence of realization must be either $D$, or a permutation of $D$ (note that in the final realization, the degree of $v_i$, $d(v_i)$, need not be equal to $d_i$).

Note that for any explicit realization, even in the simpler $\onencc$ model, there is a lower bound of $\Omega(\Delta)$ rounds (see Theorem \ref{theorem:explicit-lower-bound}). However, if we allow a $(1+\epsilon$)-factor approximation where $0\leq \epsilon \leq 1$, 
 we show 
 that in the $\onencc$ model, we can in fact compute a explicit degree-realization for any degree sequence $D$ requiring only a polylogarithmic number of rounds.


\subsection{A non-preassigned $(1+\epsilon)$ approximation solution for explicit degree realization in  the $\onencc$ model}
\label{subsection:1+e-approx}


Let $V= \{v_1,v_2,\ldots,v_n\}$ be the set of nodes of the given network and let  $D=(d_1,d_2,\ldots,d_n)$ be the given degree sequence. 
Assume w.l.o.g. that initially each $v_i$ stores the value $d_i$ (as in the previous cases). However, in the non-preassigned case, we only require that the final degree sequence realized is a permutation of $D$, thus the final degree $d(v_i)$ realized at $v_i$ need not equal (or depend upon) $d_i$. 

We begin by defining the notion of interval sequences as in~\cite{BCPR-interval:19, CDZ00,GGT11}. An interval sequence consists of $n$ integer intervals $[a_i,b_i]$ such that $0\le a_i\leq b_i \le n-1$, and is said to be {\em graphic/realizable} if there exists a graph with degree sequence $D'=(d_1',\ldots,d_n')$ satisfying the condition $a_i\leq d_i'\leq b_i$, for each $i \in[1,n]$. 

The key idea here is to appropriately convert the given degree sequence $D$ to an interval sequence $\cal S$ whose realization is at most an $\epsilon$-factor deviation, i.e., there exists a one-to-one mapping between the given degree sequence (when sorted) $D=(d_1,d_2,\ldots,d_n)$ and the realized degree sequence (when sorted) $D'=(d_1',\ldots,d_n')$ obtained from the interval realization of $\cal S$ such that 
$$d_i/(1+\epsilon) \leq d_i' \leq (1+\epsilon)d_i$$
for every $1\leq i \leq n$. 

By~\cite{BCPR-interval:19}, we know that given any interval sequence $\cal S=$ $([a_i,b_i],[a_2,b_2],\ldots, [a_n,b_n])$, there exists a centralized solution that computes a graphic sequence $D'$ (if one exists) realizing $\cal S$ in $O(n\log n)$ time. Thus, 
if we could broadcast an appropriate interval sequence to all nodes, they could rely on this centralized algorithm to obtain a graphic sequence $D'$ that would satisfy the required approximation property. However, broadcasting any length $n$ sequence is expensive. 

To reduce this communication cost, we divide the nodes into a small number of interval classes based on their assigned degrees and just broadcast the count of the number of nodes belonging to each class. Each node can locally create the interval sequence $\cal S$ based on the class node count.

Depending on their initially assigned degrees, nodes are divided into $O(\log_{1+\epsilon} \Delta)$ intervals (for a chosen suitably $\epsilon >0$) as follows. Each class $i$ ($1\leq i\leq \log_{1+\epsilon} \Delta$) of nodes consists of all nodes in $V$ whose degree is in the interval 
$$(a_i,b_i] ~=~ ((1+\epsilon)^i,(1+\epsilon)^{i+1}].$$
(The nodes at the boundaries join the lower class). 
That is, a node $v$ with degree $d(v)$ belongs to class $i$ iff 
\\
$(1+\epsilon)^i < d(v) \leq (1+\epsilon)^{i+1}$.  



Note that, through this division of nodes into classes, each node $v$ can be assigned an interval based on the class it lies in. Any node $v$ in the $i^{th}$ class would have a corresponding interval $[a_i,b_i]$.
The number of nodes in each class is calculated by aggregation and broadcast, which helps all the nodes to implicitly create the appropriate interval sequence $\cal S$ locally. Note that $\cal S$ would be the same across all nodes (for a given \ID-to-degree mapping, e.g., the highest \ID node gets the highest degree, which nodes could decide a-priori). Each node can then locally create a degree sequence $D'$ based on $\cal S$ using the centralized algorithm of \cite{BCPR-interval:19}.\suman{This paragraph is possibly redundant and can be removed.} We now describe the simple two step process in detail. 

\begin{enumerate}
\item  
Each node is assigned to a class $i$ ($1\leq i\leq \log_{1+\epsilon} \Delta$).  Let $I=\{I_1,I_2,\ldots,I_{(\log_{1+\epsilon} \Delta)}\}$ where $I_i$ is the number of nodes in the $i^{th}$ class, i.e, 
$$I_i ~=~ |\{ v_j \mid d(v_j) \in (a_i,b_i],~~ 1\leq j\leq n \}|.$$ 
For each $i \in[1,\log_{1+\epsilon} \Delta]$,  we calculate the value $I_i$ in parallel via aggregation to a common root (say the node at position $i$ of the graph) using $i$ as the group \ID. Since only the count of the nodes in a class is required, this can be achieved in $O(\log n)$ rounds via aggregation (see Theorem~\ref{thm:aggregation}). Then, via global broadcast, each $I_i$ is broadcast to every node of the network. As there are $\log_{1+\epsilon} \Delta$ values that need to be broadcast, this can be achieved in $O(\log_{1+\epsilon} \Delta) \cdot O(\log n) \leq O(\log^2 n)$ time. 
\item 
After each node in the network has received $I$, it locally creates the interval sequence $\cal S$, from which it computes a graphic sequence $D'=(d'_1,\ldots,d'_n)$ using \cite{BCPR-interval:19}. Since in the $\onencc$ model all node \IDs are known to every node, for a fixed \ID-to-degree mapping, the construction is explicit.   
\end{enumerate}



\begin{theorem}
There exists an $O(\log^2 n)$ round  ($1+\epsilon$)-approximation solution to the explicit-degree realization problem in the non-preassigned $\onencc$ model. 
\end{theorem}

\begin{proof}

We first prove that the realized degree sequence $D'$ is an  ($1+\epsilon$)-approximation solution for the given degree sequence $D$. For convenience, we consider these sequences in sorted order.
Consider the degree at the $j^{th}$ position $d_j$ in sequences $D$ and $D'$ $(1\leq j \leq n)$.  
Let the existence of $d_j$ contribute to the node count $I_i$ of the $i{th}$ class with a representative interval of $(a_i,b_i]$
iff $(1+\epsilon)^i < d_j \leq (1+\epsilon)^{i+1}$. The correctness of the realized degree sequence $D'$ given an interval sequence is guaranteed by \cite{BCPR-interval:19}. Note that, as $d_j$ lies within the interval $(a_i,b_i]$,
the realized value (corresponding to $d_j$) cannot increase by more than a $(1+\epsilon)$ factor. Since this is true for all $j \in (1,\dots, n)$, the overall approximation holds.

To prove the time complexity, we look at the various operations performed by the algorithm. Aggregating the count $I_i$ of the nodes belonging to a particular class $i$ requires $O(\log n)$ rounds via aggregation (see Theorem~\ref{thm:aggregation}) as messages from different nodes can be combined to update the count (and individual \IDs need not be sent).  To collect the individual counts from all $\log_{1+\epsilon} \Delta$ classes, global broadcast is used. As there are $\log_{1+\epsilon} \Delta$ values that need to be broadcast, this can be achieved in $O(\log_{1+\epsilon} \Delta) \cdot O(\log n) \leq O(\log^2 n)$ time (see Theorem~\ref{thm:global}).
\end{proof}
\section{Conclusion and Future Work} \label{sec:con}
We have initiated the study of graph realization problems in the distributed setting, and presented efficient algorithms in the NCC model for realizing overlay networks that satisfy  degree or connectivity requirements. We believe that such formal study of overlay network design may facilitate the 
design of a wide range of fundamental tools with strong guarantees. In addition to building such theoretical foundation, we hope our line of work may produce new ideas leading to practical improvements in building overlay networks. Our work also opens up a number of interesting new directions for future exploration. 
\begin{compactenum}
\item 
It is unclear if the lower bounds we have provided for $\zeroncc$ hold also for $\onencc$. If not, it will be interesting to achieve improved algorithms for $\onencc$.
\item 
Overlays need to be robust against a wide range of failures. Modeling these requirements in the form of properties that must be satisfied by the implemented network could lead to new graph realization variants that have not been studied so far. 
\item 
It may be interesting to realize graphs from suitable graph classes. We have already studied realization of trees, but we believe that realizing graphs in other classes like planar graphs, chordal graphs, etc. can also be of value.
\end{compactenum}

\section*{Acknowledgments} John Augustine and Sumathi Sivasubramaniam are supported by DST/SERB Extra Mural Grant (file number EMR/2016/00301). John Augustine is also supported by DST/SERB MATRICS Grant (file number MTR/2018/001198). David Peleg is supported by the Venky Harinarayanan and Anand Rajaraman Visiting Chair Professorship at IIT Madras. Supported by the chair's funds, the work was done in part when David Peleg, Avi Cohen, and Keerti Choudhary visited IIT Madras and when John Augustine visited Weizmann Institute of Science. This work was done in part when Suman Sourav visited IIT Madras. 
A short version of this paper appears in the proceedings of IPDPS 2020.

\end{document}